\documentclass[%
 reprint,
superscriptaddress,
 amsmath,amssymb,
 aps,
pra,
showkeys
]{revtex4-2}

\usepackage{graphicx}% Include figure files
\usepackage{dcolumn}% Align table columns on decimal point
\usepackage{bm}% bold math
\usepackage{algorithm}%
\usepackage{algorithmicx}%
\usepackage[noend]{algpseudocode}%

\usepackage{amsthm}
\usepackage{amssymb}
\usepackage{amsmath}
\usepackage{bm}
\usepackage{subfigure}
\usepackage{siunitx}
\usepackage{footmisc}
\DeclareMathOperator{\sgn}{sgn}

\newtheorem{theorem}{Theorem}%  meant for continuous 
\newtheorem{corollary}{Corollary}%
\newtheorem{lemma}{Lemma}%
\renewenvironment{proof}{\paragraph*{Proof:}}{\hfill$\square$}

\begin{document}

\preprint{APS/123-QED}

\title{Analytical results for the Quantum Alternating Operator Ansatz with Grover Mixer}% Force line breaks with \\

\author{Guilherme Adamatti Bridi}
\email{gabridi@cos.ufrj.br}
\affiliation{%
 Department of Computer Science and Systems Engineering, Federal University of Rio de Janeiro, Rio de Janeiro 22290-240, Brazil
}%

\author{Franklin de Lima Marquezino}%
\email{franklin@cos.ufrj.br, franklin.marquezino@lu.lv}
\affiliation{%
 Department of Computer Science and Systems Engineering, Federal University of Rio de Janeiro, Rio de Janeiro 22290-240, Brazil
}%
\affiliation{%
 Duque de Caxias Campus, Federal University of Rio de Janeiro, Duque de Caxias 25240-005, Brazil
}%
\affiliation{%
 Center for Quantum Computer Science, University of Latvia, Riga LV-1586, Latvia
}%

\date{\today}% It is always \today, today,

\begin{abstract}
The Grover mixer operator is a variational version of Grover's diffusion operator, introduced as a mixing operator for the Quantum Alternating Operator Ansatz (QAOA) and used in a variant known as Grover Mixer QAOA (GM-QAOA). An important property of QAOA with Grover mixer is that its expectation value is invariant over any permutation of states. As a consequence, the algorithm is independent of the structure of the problem. If, on the one hand, this characteristic raises serious doubts about the capacity of the algorithm to overcome the bound of the unstructured search problem, on the other hand, it can pave the way to its analytical study. In this sense, a prior work introduced a statistical approach to analyze GM-QAOA that results in an analytical expression for the expectation value depending on the probability distribution associated with the problem Hamiltonian spectrum. Although the method provides surprising simplifications in calculations, the expression depends exponentially on the number of layers, which makes direct analytical treatment unfeasible. In this work, we extend the analysis to the more simple context of Grover Mixer Threshold QAOA (GM-Th-QAOA), a variant that replaces the phase separation operator of GM-QAOA to encode a threshold function. As a result, we obtain an expression for the expectation value independent of the number of layers and, with it, we provide bounds for different performance metrics. Furthermore, we extend the analysis to a more general context of QAOA with Grover mixer, which we called Grover-based QAOA. In that framework, which allows the phase separation operator to encode any compilation of the cost function, we generalize all the bounds by using an argument by contradiction with the optimality of Grover's algorithm on the unstructured search problem. As a result, we get the main contribution of this work, an asymptotic lower bound on the quantile achieved by the expectation value that formalizes the notion that the Grover mixer, at most, reflects a quadratic Grover-style speed-up over classical brute force. We apply that bound on the Max-Cut problem to the particular class of complete bipartite graphs and argue that the number of rounds required to achieve guarantees for any approximation ratio must grow exponentially with the number of vertices/edges, a severe limitation on the performance of the algorithm.
\end{abstract}

%\keywords{QAOA, Grover mixer, GM-QAOA, GM-Th-QAOA, Threshold QAOA, quadratic speed-up, probability distributions}%Use showkeys class option if keyword
\maketitle

%\tableofcontents

\section{Introduction} \label{sec1}

In the current paradigm of quantum computing technology, known as the Noisy Intermediate-Scale Quantum (NISQ) era~\cite{nisq}, we only have at our disposal a few noisy qubits and no error correction. In the context of quantum optimization~\cite{quantum_opt}, the Variational Quantum Algorithms (VQA)~\cite{vqa}, a class of hybrid quantum-classical algorithms, have gained prominence in recent years for the NISQ scenario. These algorithms work with parameterized quantum circuits of limited depth and number of qubits, using classical procedures to optimize the optimization parameters.

One of the most prominent cases of VQA is the Quantum Approximate Optimization Algorithm (QAOA)~\cite{qaoa}, which can be generalized to the Quantum Alternating Operator Ansatz (QAOA)~\cite{qaoaAnsatz}. QAOA is a class of algorithms derived from the Quantum Adiabatic Algorithm~\cite{qaa1, qaa2} and used heuristically to find solutions to combinatorial optimization problems. The algorithm consists of a given number of rounds of alternating application of two parametrized operators in an initial state. The first is the phase separation operator, which changes the relative phases between states by introducing bias according to the cost function (the objective function) to be optimized. The last one is the mixing operator, responsible for generating interference between the states, changing its amplitudes with the goal of amplifying the states corresponding to high-quality solutions. 

The original mixing operator of QAOA uses the transverse field mixer Hamiltonian~\cite{qaoa}, which is given by a sum of Pauli-$X$ operators. Since then, many other variations with different types of mixers have already been introduced in the literature~\cite{qaoaAnsatz, GM-QAOA, qwoaRestriction, qwoaCombinatorial, constraintMixer}. There is numerical evidence that the choice of mixing operator significantly affects the performance of QAOA~\cite{reachability, numericalEvidence1, numericalEvidence2} and therefore choosing the ideal mixer for a given optimization problem is an important research topic. One variant of particular interest is the Grover Mixer Quantum Alternating Operator Ansatz (GM-QAOA), introduced for both unconstrained~\cite{reachability, GM_unconstrained, slightlyGrover} and constrained~\cite{GM-QAOA} optimization context. The mixing operator of GM-QAOA is a variational version of Grover's diffusion operator~\cite{grover1, grover2} called Grover mixer operator. In GM-QAOA formulation, a necessary condition to the construction of the mixing operator is the existence of an efficient preparation of uniform superposition over the feasible states---which covers constrained problems like the Traveling Salesman Problem, the Max $k$-Vertex Cover, and the Discrete Portfolio Rebalancing. Alternatively, the operator of Grover mixer can be constructed with the formulation of Quantum Walk-based Optimization Algorithm (QWOA)~\cite{qwoaRestriction, qwoaCombinatorial}, a generalization of QAOA that interprets the mixing operator as a continuous-time quantum walk operator~\cite{CTQW1, CTQW2, QW_renato}. In that case, the Grover mixer operator is equivalent to QWOA on the complete graph up to a change of scale on the operator parameter. Problems such as the Capacitated Vehicle Routing~\cite{vehicle} and the Portfolio Optimization~\cite{portfolio} have already been numerically studied within the QWOA framework.

Another variant of Quantum Alternating Operator Ansatz using the Grover mixer is the so-called Grover Mixer Threshold QAOA (GM-Th-QAOA)~\cite{GM-Th-QAOA}, an algorithm combing such mixing operator with the more general Threshold QAOA, which in turn changes the original phase separation to encode a compilation of the cost function into a threshold function splitting the solution space from a value. That variant is closely related to Grover's algorithm. In the general context of QAOA, Jiang, Rieffel, and Wang~\cite{grover_tranverseField} show that the algorithm with original transverse field mixer can achieve the same asymptotic quadratic speed-up over classical brute force as the Grover's algorithm for the unstructured search problem. In the particular case of GM-Th-QAOA, the relationship is much more direct since the choice of all angles as being equal to $\pi$ reduces the algorithm to a Grover's search for marked states above (considering the original definition) a given threshold. An advantage of this variant is admitting an efficient procedure to find optimal parameters---the angles and the threshold value---that eliminates the costly outer loop parameter finding of the usual QAOA. Furthermore, it has been numerically observed that the performance of GM-Th-QAOA consistently overcomes GM-QAOA in all instances considered~\cite{GM-Th-QAOA, numericalEvidence1, numericalEvidence2}.

The performance of QAOA with the Grover mixer, individually or compared with other mixers, has already been considered in the literature. The initial thought, corroborated by numerical experiments on small instances, was that the Grover mixer would overcome transverse field mixer due to its ability to mix quickly and its global symmetry among states~\cite{reachability, qwoaCombinatorial, portfolio}. However, later experiments on larger instances indicated that this advantage soon disappeared with Grover mixer losing to transverse field mixer on unconstrained context~\cite{numericalEvidence1} and performing even exponentially worse than the clique mixer~\cite{XY1, XY2} on constrained problems~\cite{numericalEvidence2}. One can argue that the worst performance of the Grover mixer may be due to the fact that it depends only on the distribution of the solution space---by the global symmetry, the Grover mixer is invariant under any permutation of states~\cite{headley_paper}---so the algorithm does not see the structure of the optimization problem and is possibly limited to the bound of the unstructured search problem~\cite{asymptBoundGrover} (i.e., a Grover-style quadratic speed-up), drastically compromising algorithm performance on large instances. On the other hand, other mixers, such as transverse field and clique, could, in principle, overcome that limit by exploiting the underlying problem structure.

Despite the performance limitation, the Grover mixer provides a unique opportunity to get analytical studies for QAOA. Historically, analytical results are rare and sparse in QAOA literature due to the high complexity of quantum operators (see some examples on Ref.~\cite{qaoa, maxcut_ar1, maxcut_ar2, sherrington-kirkpatrick, XY1, Vijendran}). However, the independence of the structure of the Grover mixer can greatly simplify the analysis. That has been noticed by Bennett and Wang~\cite{maoa}, who used degeneracy in solution space to make edge contractions on the complete graph of QWOA. Headley and Wilhelm~\cite{headley_paper} went further and introduced a statistical approach using random variables to model the problem that led to an analytical expression of the expectation value of GM-QAOA depending only on the probability distribution associated with the problem Hamiltonian spectrum---the solution space of the optimization problem. The prominent statistical quantity of the resulting expression is the characteristic function, i.e., the Fourier transform of the probability distribution. Although the complexity of that expression scales exponentially with $\mathcal{O}(4^r)$, where $r$ is the number of layers, it does not depend on the size of the problems, which allows computing the optimal parameter (or near-optimal) in size limit for problems with instances that converge asymptotically towards a fixed distribution, such as the Number Partition Problem with i.i.d. choice of the numbers. The statistical approach was later generalized by Headley~\cite{headley_thesis} to the structure-dependent transverse field mixer and the so-called line mixer, in a work that establishes how QAOA can actually exploit the underlying structure of the problem.

\subsection{Our contributions} \label{sub1A}

Although the Headley and Wilhelm~\cite{headley_paper} method provides surprising simplifications in GM-QAOA expectation value calculations, the expression is still too complicated to obtain formal bounds on the algorithm's performance through direct analytical treatment. In the present work, with the motivation of understanding the theoretical potential of the Grover mixer on combinatorial optimization---mainly to investigate the issue concerning the quadratic Grover-like speed-up---we extended the analysis to the more simple case of GM-Th-QAOA. Using the well-known formula of the probability of Grover's algorithm and its optimality on average probability for the unstructured search~\cite{singleOptimalGrover, generalOptimalGrover}, we provide an expression for the expectation value in Theorem~\ref{thm1} with complexity independent of the number of layers, which allows us to study the asymptotic behavior of the algorithm. Rather than the characteristic function, the prominent statistical quantity here is the conditional expected value. With a closed-form expression, in Theorem~\ref{thm2}, we prove the conjecture on which the efficient method of Golden et al.~\cite{GM-Th-QAOA} of finding the optimization parameters is based (see details on Subsec.~\ref{sub3A}). Furthermore, we provide bounds on the performance of the expectation value of GM-Th-QAOA with the statistical quantities of quantile and the standard score. On the first, we get in Theorem~\ref{thm3} an asymptotic tight bound that implies the expected quadratic Grover-style speed-up of GM-Th-QAOA over classical brute force, the most relevant contribution of this work to the specific context of GM-Th-QAOA. On the second, from Lemma~\ref{lm1}, we conclude that the maximum standard score achieved by the expectation value of GM-Th-QAOA is hit by binary functions with specific ratios, and investigating the performance of these distributions, we state in Theorem~\ref{thm4} that the standard score scales at most linearly with the number of rounds. As an immediate consequence, we bound in Corollary~\ref{crlr3} the minimum number of layers to achieve guarantees for a fixed approximation ratio. Finally, we combine both bounds to argue that the algorithm's performance is closely related to the asymptotic behavior of the probability distribution on the limit of its support.

To get stronger results about the Grover mixer, we also consider a more general version of QAOA with the Grover mixer that we called \textit{Grover-based QAOA}. In that framework, which includes both GM-QAOA and GM-Th-QAOA, is allowed that the phase separation operator codifies any compilation of the cost function in a real-valued function. We generalize all bounds of GM-Th-QAOA to Grover-based QAOA with a technique used in Lemma~\ref{lm2} that consists of bounding the maximum amplification of the probability of measuring a set of degenerate states. The argument used for this consists of showing that if there is an amplification of states greater than the one provided by Grover's search, i.e., a quadratic amplification on the probability of sampled states, we contradict the optimality of average provability on unstructured search~\cite{singleOptimalGrover, generalOptimalGrover}, building an explicit algorithm that performance better than Grover's algorithm. With that limitation on the probability amplification of states, we get in Theorem~\ref{thm5} a bound on the expectation value of any Grover-based QAOA, and from it, we conclude that all bounds on Grover-based QAOA follow the same asymptotic behavior as its correspondents on GM-Th-QAOA. Specifically, Corollary~\ref{crlr4} generalize Theorem~\ref{thm3} and Theorem~\ref{thm6} is the analogous for Grover-based QAOA of both Theorem~\ref{thm4} and Corollary~\ref{crlr3}. In particular, Corollary~\ref{crlr4} gives the main general contribution of this work, the formalization of the notion that the Grover mixer is limited to the quadratic speed-up over classical brute force with an asymptotic performance, for instance, analog to the Grover Adaptive Search (GAS)~\cite{gas1, gas2, gas3, gas4}. We apply the bounds in the context of the Max-Cut problem. That way, by using the knowledge of the asymptotic behavior of the probability distribution associated with the particular case of complete bipartite graphs, we argue that for this class of graphs, the number of rounds required to achieve a fixed approximation ratio must grow exponentially with the number of vertices/edges, a severe limitation on the performance of the Grover mixer. More than that, the construction suggests that it is likely that the same happens with other classes of graphs and even with other combinatorial optimization problems.

The structure of the paper is as follows. In Sec.~\ref{sec2}, we formally define the Grover-based QAOA and introduce random variables with statistical concepts necessary to the analysis. In Sec.~\ref{sec3}, we present the analytical results of GM-Th-QAOA. In Sec.~\ref{sec4}, we provide numerical experiments involving probability distributions to emphasize and illustrate important aspects of the results of Sec.~\ref{sec3}. In Sec.~\ref{sec5}, we generalize the analytical bounds of GM-Th-QAOA to Grover-based QAOA and apply it to the Max-Cut problem. In Sec.~\ref{sec6}, we present our conclusions.

\section{Definitions} \label{sec2}

\subsection{The Grover-based QAOA} \label{sub2A}

We define Grover-based QAOA as follows. Consider an instance of a combinatorial optimization problem defined on a domain $S$ with a cost function $c(k): S \rightarrow \mathbb{R}$ to be minimized. The goal of Grover-based QAOA is to minimize the expectation value of the Hamiltonian problem $H_C$ that encodes the cost function $c(k)$ on the parameterized state ${\displaystyle | \psi^{(r)} \rangle}$. The algorithm acts in some feasible subspace of a Hilbert space of a qubit system spanned by $M = |S|$ computational basis states codifying the solutions of $S$ and ${\displaystyle | \psi^{(r)} \rangle}$ is given by
\begin{equation}
    \label{eqBa01}
    {\displaystyle | \psi^{(r)} \rangle} = U_M(\beta_r) U_P(\gamma_r) \ldots U_M(\beta_1) U_P(\gamma_1) {\displaystyle | s \rangle}.
\end{equation}
Here,
\begin{itemize}
    \item $r$ is the number of rounds/layers or the depth of QAOA.
    \item ${\displaystyle | s \rangle}$ is a uniform superposition over all states of $S$;
    \item $U_M(\beta) = e^{i \beta H_M}$ is the mixing operator, where $H_M$ is Grover mixer Hamiltonian, given by $H_M =  {\displaystyle | s \rangle} {\displaystyle \langle s |}$;
    \item $U_P(\gamma) = e^{i \gamma H_Q}$ is the phase separation operator, where $H_Q$ is a diagonal Hamiltonian that encodes a real-valued function $q(k)$ compiled from the cost function such that $H_Q {\displaystyle | k \rangle} = q(k) {\displaystyle | k \rangle}$;
    \item sets $\bm{b} = (\beta_1, \ldots , \beta_r)$ and $\bm{\gamma} = (\gamma_1, \ldots , \gamma_r)$ are the optimization parameters (or angles).
\end{itemize}

For GM-QAOA, the function $q(k)$ is precisely the cost function, and for GM-Th-QAOA, we compile $q(k)$ from $c(k)$ by a threshold function given by
\begin{equation}
    \label{eqBa02}
    T_h(k) = 
    \begin{cases}
        -1, &  c(k) \leq t\\
        0,              & \text{otherwise},
    \end{cases}
\end{equation}
for a threshold value $t$ that must be optimized. The minus sign on Eq.~\eqref{eqBa02} is because in this work we consider QAOA for minimization problems. The expectation value of GM-QAOA and GM-Th-QAOA of depth $r$ is denoted here by $E_r(\bm{\beta}, \bm{\gamma})$ and $E_r(t)$~\footnote{Despite the dependency of the expectation value concerning $\bm{\beta}$ and $\bm{\gamma}$, since we get the optimal angles at the beginning of the Sec.~\ref{sec3}, we simplified the notation from the outset.}, respectively, while $E_r$ denotes the expectation value of a generic Grover-based QAOA. 

Furthermore, although QAOA primarily must be proper to NISQ devices in such a way that the number of layers must be low, in that work, we consider the asymptotic limit of $r$ to several analytical results of the Sections~\ref{sec3} and~\ref{sec5}, as well as we simulate large values of $r$ for the numerical experiments of Sec.~\ref{sec4}. The main reason for this study is to address the issue of the asymptotic quadratic of Grover-like speed-up on the variants of QAOA with the Grover mixer.

\subsection{Random variables} \label{sub2B}

In this work, we apply the Headley and Wilhelm~\cite{headley_paper} approach, introduced for GM-QAOA, to the GM-Th-QAOA. That approach consists of using random variables as a mathematical model, writing the analytical expression of the expectation value in terms of the probability distribution associated with the solution space. If the probability distribution of a combinatorial optimization problem is known or can even be approximated, we can analytically optimize the angles $\bm{\beta}$ and $\bm{\gamma}$. That way, the expectation value can be computed directly from sampling the quantum circuit, avoiding the costly outer loop optimization procedure. For particular combinatorial optimation problems, Headley and Wilhelm~\cite{headley_paper} analytically conclude that the probability distribution of the Number Partition Problem with i.i.d. choice of the numbers follows a chi-squared distribution. On the other hand, in the empirical sense, it was observed that the solution space of the Capacitated Vehicle Routing and Portfolio Optimization problems seems to be normally distributed~\cite{vehicle, maoa}.

In particular, Headley and Wilhelm~\cite{headley_paper} use continuous random variables to model the solution space as an asymptotic approximation on the large size limit. We, in contrast, use the exact case of discrete random variables. Formally, let $X$ be the random variable of uniformly sampling an element on the set $S$ and calculating the cost function. The function 
\begin{equation}
    \label{eqBb01}
  f_X(x) = \frac{| \{ k \in S: c(k) = x \} |}{M}
\end{equation}
is the probability mass function (pmf) of $X$, and the support $R_X$ of $X$ is a countable subset of real numbers. We denote the mean, standard deviation, minimum value, and maximum value of $X$ by $\mu = \operatorname{E}[X]$, $\sigma = \sqrt{\operatorname{E}[X - \mu]^2}$, $R_X^{min}$, and $R_X^{max}$, respectively, and assume $0 < \sigma < \infty$ (ignoring the degenerate distribution). Provided that $R_X^{min} \neq 0$ and $|R_X^{min}| < \infty$, the approximation ratio is defined by
\begin{equation}
    \label{eqBb02}
  \lambda = \frac{E}{R_X^{min}}, \ E \in \{E_r(\bm{\beta}, \bm{\gamma}), E_r(t), E_r \},
\end{equation}
that is, $E$ is the expectation value of the considered QAOA variant in each context. That metric, extensively used in the context of combinatorial optimization, measures the proximity of the solution provided by the algorithm to the optimal solution. We also define the random variable $Y$ and the standard random variable $Z$ as
\begin{equation}
    \label{eqBb03}
  Y = X - \mu, \ Z = \frac{X - \mu}{\sigma}.
\end{equation}

For GM-QAOA, the summations of the original expectation value expression are replaced by the characteristic function of $X$ and its derivative (see Eq.~\ref{eqCc03} for the resulting analytical expressions). In contrast, for GM-Th-QAOA, due to the binary nature of the function $q(k)$ on that variant, the main statistical quantity of the analysis of GM-Th-QAOA is the conditional expectation. The expectation of $X$ given $X \leq x$ and the expectation of $X$ given $X > x$ are
\begin{equation}
    \label{eqBb04}
   \operatorname{E}[X | X \leq x] = \frac{G_X(x)}{F_X(x)}, \ \operatorname{E}[X | X > x] = \frac{\mu - G_X(x)}{1 - F_X(x)},
\end{equation}
respectively, where 
\begin{equation}
    \label{eqBb05}
   G_X(x) = \sum_{k \in R_X: k \leq x} k f_X(k), \ F_X(x) = \sum_{k \in R_X: k \leq x} f_X(k).
\end{equation}
The function $F_X(x)$ is the cumulative distribution function (cdf). Note that with these definitions, we can consider any $x \in \mathbb{R}$ as argument. We replace, on the original expectation value expressions of GM-Th-QAOA, the statistical quantities $F_X(x)$ and $G_X(x)$, where the argument $x$ is the threshold value $t$, getting an expression that is dependent on the ratio between the number of states below or equal to the threshold and the total number of states.

Furthermore, in part of the present work, $F_X(x)$ and $G_X(x)$ must be differentiable. However, $X$ was defined to be discrete. To bypass that obstacle, note that we can write $F_X(x)$ and $G_X(x)$ as $F_X(x) = \sum_{k \in R_X} f_X(k) \theta(x - k)$ and $G_X(x) = \sum_{k \in R_X} k f_X(k) \theta(x - k)$, where $\theta(x)$ is the Heaviside step function. The derivatives of $F_X(x)$ and $G_X(x)$ are then given by $f_X^G(x)$ and $x f_X^G(x)$ respectively, where $f_X^G(x) = \sum_{k \in R_X} f_X(k) \delta(x - k)$, with $\delta(x)$ as the Dirac delta function. The distribution $f_X^G(x)$ is an extension of the concept of probability density function for non-continuous random variables, usually called generalized probability density function~\cite{deltaFunction}.

In other situations, on the other hand, it is convenient for $F_X(x)$ to be continuous, which does not occur for generalized density functions. That way, we consider $X$ as a continuous random variable with an asymptotic approximation, in the same way as Headley and Wilhelm~\cite{headley_paper}. Consequently, all summations are replaced by integrals, and the probability mass function becomes a probability density function (pdf). However, beyond mere convenience, this assumption is justifiable since the main QAOA target is NP-Hard optimization problems with sufficient scalability to make it a good approximation.

\section{Grover Mixer Threshold QAOA analysis} \label{sec3}

An important component of our GM-Th-QAOA analysis is the result of GM-QAOA for a binary function defined here, taking the value $-1$ for elements belonging to a subset of marked elements and $0$ otherwise. We denote by $\rho$ the ratio of marked elements to the entire domain of the function. According to the statistical interpretation of GM-QAOA by Headley and Wilhelm~\cite{headley_paper}, the distribution of the binary function follows a reflected Bernoulli distribution.

The expectation value of $r$ rounds of GM-QAOA for the binary function is the negative of the probability of measuring a marked state. It is important to note that GM-QAOA, applied to this particular input problem, is equivalent to the unstructured search problem with an arbitrary number of marked elements. This equivalence arises from the fact that there are $r$ calls to an oracle for the binary function, and the objective of minimizing the expectation value aligns to maximize the probability of measuring a marked state.

Notice also that minimizing the expectation value GM-Th-QAOA with fixed threshold $t$ is equivalent to minimizing GM-QAOA with the binary function as input, and consequently to maximize the probability of measuring an element smaller or equal to the threshold. Not by coincidence, for $r = 1$, we know by Golden et al.~\cite{GM-Th-QAOA} that when $\rho \leq 0.25$, the optimal angles $\beta = \gamma = \pi$ reduces a GM-QAOA round to a Grover's iteration. Combining it with the initial condition of uniform superposition over all states, we have an emulation of Grover's algorithm. Indeed, $\rho = 0.25$ is the ratio in which Grover's algorithm with a single round reaches the probability $1$ on measuring a marked element, so that if $\rho \leq 0.25$, Grover's operators are optimal, and if $\rho > 0.25$, the angles $\beta = -\gamma = \arctan{(-\sqrt{4\rho - 1}, 2\rho - 1)}$ make the fine-tuning not to exceed the point of probability $1$.

The aforementioned ratio that reaches probability $1$, named here \textit{threshold ratio} and denoted as $\rho_{Th}(r)$ for arbitrary $r$, can be seen as the point of $\pi/2$ radians angle of the geometric interpretation of Grover's algorithm. The value of the threshold ratio is $\rho_{Th}(r) = \sin^2{(\pi/(4r + 2))}$. Up to this point, for any number of iterations, Grover's algorithm gives the maximal average probability for measuring a marked state on the unstructured search problem. The proof was done by Zalka~\cite{singleOptimalGrover} for a single marked element and generalized for an arbitrary number of marked elements by Hamann, Dunjko, and Wölk~\cite{generalOptimalGrover}. Note that, for instance, the variational approach done by Morales, Tlyachev, and Biamonte~\cite{slightlyGrover} performs slightly better than Grover's algorithm because the marked element ratio of the instances surpassed $\rho_{Th}(r)$.

Denoting by $P(\rho, r)$ the optimal probability of measuring a marked element, we can generalize GM-QAOA performance on binary function for an arbitrary number of layers to
\begin{equation}
    \label{eqC01}
P(\rho, r) = 
    \begin{cases}
        \sin^2{((2r + 1) \arcsin{(\sqrt{\rho}}))}, &  \rho \leq \rho_{Th}(r)\\
        1,              & \text{otherwise}.
    \end{cases}
\end{equation}

The first interval is the closed-form expression for Grover's algorithm probability derived due to its optimality on average probability, applicable to GM-QAOA since the expectation value of QAOA with Grover mixer operator is invariant under any permutation of states---i.e., the positions of marked elements~\cite{headley_paper}. To establish the other interval, we split into $\rho \leq \rho_{Th}(r - 1)$ and $\rho > \rho_{Th}(r - 1)$ cases. To the first, we set $\beta_j = \gamma_j = \pi$ for all $j < r$ and
\begin{equation}
    \label{eqC02}
    \begin{split}
    & \beta_r = \arctan{\left(-\sqrt{\Delta} |c_{0, \pi}^{(r - 1)}| ,  \frac{2\rho}{M} - (c_{0, \pi}^{(r - 1)})^2 \right)},
    \\ & \gamma_r = -\arctan{\left(-\frac{\sqrt{\Delta}}{c_{1, \pi}^{(r - 1)} \sgn{(c_{0, \pi}^{(r - 1)})}} , \frac{c_{0, \pi}^{(r - 1)}(2r - 1)}{c_{1, \pi}^{(r - 1)}} \right)},
    \end{split}
\end{equation}
where $\Delta = 4\rho/M - (c_{0, \pi}^{(r - 1)})^2$. Here, $c_{0, \pi}^{(r - 1)}$ and $c_{1, \pi}^{(r - 1)}$ are the amplitudes of non-market and marked states, respectively, on our $\pi$ attribution of parameters after $r - 1$ rounds. By Golden et al.~\cite{GM-Th-QAOA} analysis, $P(\rho, r) = 1$ is achieved for $\Delta > 0$. Our analysis range corresponds to the proved interval of Eq.~\eqref{eqC01} on $P(\rho, r - 1)$. Therefore, writing $(c_{0, \pi}^{r - 1})^2$ as a function of $P(\rho, r - 1)$ on the expression of $\Delta$ gives
\begin{equation}
    \label{eqC03}
    \Delta = \frac{1}{M}\left(4\rho - \frac{1 - P(\rho, r - 1)}{1 - \rho} \right).
\end{equation}
The value of $\Delta$ is increasing on our range and therefore is enough to prove that $\Delta = 0$ at the point $\rho = \rho_{Th}(r)$. Thus, applying trigonometric identities $2\sin^2{(x)} = 1 - \cos{(2x)}$ and $2\cos^2{(x)} = 1 + \cos{(2x)}$ on the Eq.~\eqref{eqC03} at that point results
\begin{equation}
    \label{eqC04}
\Delta = \frac{1}{2M(1-\rho)} \left(\cos{\left(\frac{2 \pi(2r - 1)}{4r + 2}\right)} + \cos{\left(\frac{4 \pi}{4r + 2}\right)}\right).
\end{equation}
The equality of Eq.~\eqref{eqC04} with $0$ follows from the fact that the sum of the arguments of both cosines is $\pi$ for any $r$. To finish, in the case where $\rho > \rho_{Th}(r - 1)$, there exists $k$ such that $1 \leq k < r$ on which $\rho_{Th}(k - 1) \geq \rho > \rho_{Th}(k)$. Probability $1$ can be reached with the earlier attribution on the $k$th first layers and $\beta_j = \gamma_j = 0$ to the remainder parameters to make the operators trivial. 

Note that by proving the optimality of the choice of all the angles on $\bm{\beta}$ and $\bm{\gamma}$ being equal to $\pi$ on $\rho \leq \rho_{Th}(r)$ interval, we prove that the efficient method of parameter finding of Golden et al.~\cite{GM-Th-QAOA} indeed finds the optimal angles for a fixed $t$. Established Eq.~\eqref{eqC01}, Theorem~\ref{thm1} explicitly expresses the expectation value of GM-Th-QAOA, assuming the discussed optimal angles on binary function, in terms of $P(\rho, r)$, in addition to statistical quantities of the random variables $X$ and $Y$.

\begin{theorem} \label{thm1}
For any number $r$ of layers in GM-Th-QAOA with optimal angles, the expectation value is given by
\begin{equation}
    \label{eqC05}
E_r(t) = \mu - G_Y(T) \frac{1 - P(\rho, r)/F_Y(T) } {1 - F_Y(T)},
\end{equation}
where $T = t - \mu$ and $P(\rho, r)$ is the optimal probability of measuring a marked solution in the binary function of ratio $\rho = F_Y(T)$ for GM-QAOA, given by Eq.~\eqref{eqC01}. For $F_Y(T) = 0$ and $F_Y(T) = 1$, we consider the respective limits on Eq.~\eqref{eqC05}.
\end{theorem}

\begin{proof}
The final state of GM-Th-QAOA for arbitrary $r$ can be expressed using arbitrary amplitudes $c_1^{(r)}$ for the set of states below or equal the threshold, and $c_0^{(r)}$ for the set of states above the threshold,
\begin{equation}
    \label{eqC06}
{\displaystyle | \psi^{(r)} \rangle} = c_1^{(r)} \sum_{k \in S: c(k) \leq t} {\displaystyle | k \rangle} + c_0^{(r)} \sum_{k \in S: c(k) > t} {\displaystyle | k \rangle},
\end{equation}
and the expectation value calculated by 
\begin{equation}
    \label{eqC07}
E_r(t) = |c_1^{(r)}|^2 \sum_{k \in S: c(k) \leq t} c(k) + |c_0^{(r)}|^2 \sum_{k \in S: c(k) > t} c(k).
\end{equation}
These summations can be performed equivalently using the pmf of $X$. For each possible cost, $x \in R_X$, we count the number of solutions $k$ such that $c(k) = x$, i.e., $M f_X(x)$. Then,
\begin{equation}
    \label{eqC08}
    \begin{split}
E_r(t) & = M |c_1^{(r)}|^2 \sum_{x \in R_X: \ x \leq t} x f_X(x) \\ & + M |c_0^{(r)}|^2 \sum_{x \in R_X: \ x > t} x f_X(x).
\end{split}
\end{equation}
Let $m$ be the number of states smaller or equal to $t$. We assume $0 < m < M$. Using $F_X(t) = m/M$, the definition of $G_X(t)$ and noting that the probability $P(\rho, r)$ of measuring a state smaller or equal to $t$ is $m |c_1^{(r)}|^2$, we can rewrite the expression as
\begin{equation}
    \label{eqC09}
E_r(t) = \frac{G_X(t)} {F_X(t)} P(\rho, r) + \frac{\mu - G_X(t)}{1 - F_X(t)} (1 - P(\rho, r)).
\end{equation}
The random variable $Y$ is introduced on our expression by using the properties $F_X(t) = F_Y(T)$ and $G_X(t) = \mu F_Y(T) + G_Y(T)$. Thus,
\begin{equation}
    \label{eqC10}
E_r(t) = \mu + \frac{G_Y(T)} {F_Y(T)} P(\rho, r) - \frac{G_Y(T)}{1 - F_Y(T)} (1 - P(\rho, r)),
\end{equation}
and with some algebraic manipulations, Eq.~\eqref{eqC05} follows. Now, note that the limit of $E_r(t)$ on $F_Y(T) = 0$ or $F_Y(T) = 1$ is $\mu$, the desired value since both cases represent a uniform superposition on the final GM-Th-QAOA state. 
\end{proof}

Some aspects and consequences of the above theorem are worth commenting on. Firstly, for a fixed $t$, we have an expression for the expectation value $E_r(t)$ with complexity independent of $r$. That allows us to analyze distributions with an arbitrary number of layers, far beyond the Headley and Wilhelm~\cite{headley_paper} approach of GM-QAOA, and look at the asymptotic behavior on the number of layers.

Proceeding, by Eq.~\eqref{eqC09} and the definition of conditional expectation,
\begin{equation} 
\label{eqC11}
E_r(t) = \operatorname{E}[X | X \leq t] P(\rho, r) + \operatorname{E}[X | X > t] (1 - P(\rho, r)).
\end{equation}
The above equation gives an important intuition on the operation of GM-Th-QAOA. The expectation value $E_r(t)$ is a weighted sum by $P(\rho, r)$ of the expected value of the two sets split by the threshold value $t$. The intuitive notion that $E_r(t) < \mu$ on $0 < F_Y(T) < 1$, holds from $P(\rho, r)/F_Y(T) > 1$, $G_Y(T) < 0$. The ratio $P(\rho, r)/F_Y(T)$, denoted by $\eta$, represents the amplification of the probability of measuring a marked state after the application of the operators.

Moreover, $E_r(t)$ admits to be written as a polynomial function in terms of $F_Y(T)$ on $F_Y(T) \leq \rho_{Th}(r)$ interval. Using the trigonometric identities
\begin{equation}
    \label{eqC12}
    \sin{(nx)} = \sum_{k = 0}^{\lfloor \frac{n - 1}{2} \rfloor} (-1)^k \binom{n}{2k + 1} \sin^{2k + 1}{(x)} \cos^{n - 2k - 1}{(x)},
\end{equation}
and $\cos{(\arcsin{(x)})} = \sqrt{1 - x^2}$, we can express $P(\rho, r)$ as
\begin{equation} 
    \label{eqC13}
    P(\rho, r) = 
     \rho \left(\sum_{k = 0}^{r} (-1)^k \binom{2r + 1}{2k + 1} \rho^k (1 - \rho)^{r - k}  \right)^2.
\end{equation}
Plugging it into Eq.~\eqref{eqC05} gives for $E_r(t)$ a polynomial expression on $F_Y(T)$ with order $2r - 1$. In particular, for $r = 1$, $P(\rho, 1) = \rho(4 \rho - 3)^2$ and then $E_1(t) = \mu + 8 G_Y(T)(1 - 2 F_Y(T))$.

Now, if $F_Y(T) \geq \rho_{Th}(r)$ then $P(\rho, r) = 1$, and immediately
\begin{equation}
    \label{eqC14}
    \begin{split}
E_r(t) & = \mu + \frac{G_Y(T)}{F_Y(T)} = \mu + E\left[Y| Y \leq T\right] \\ & = E\left[X| X \leq t \right],
    \end{split}
\end{equation}
which induces a tight lower bound on GM-Th-QAOA performance and an upper bound on $t_{opt}$, both given  by Corollary~\ref{crlr1}.

\begin{corollary} \label{crlr1}
For any number $r$ of layers in GM-Th-QAOA, the optimal threshold value is bounded by $t_{opt} \leq \tau$, where $\tau$ is the minimum $t$ in which $P(\rho, r) = 1$, and we have the tight bound in the optimal expectation value given by
\begin{equation}
    \label{eqC15}
E_r(t)_{opt} \leq E\left[X| X \leq \tau \right] \leq \tau.
\end{equation}
\end{corollary}

\begin{proof}
First, note that there always exists a $t$ for which $P(\rho, r) = 1$ because taking $t$ as the maximum solution gives $F_X(t) = 1$. Then, by the definition of conditional expectation, the minimum $t$ gives the best expectation value among the candidates of threshold in which $P(\rho, r) = 1$. To conclude the tightness of the bound, we consider the binary function with a parameter $\rho$ such that $P(\rho, r) = 1$. In that case, $\rho = F_X(\tau)$ and then $E_r(t) = E\left[X| X \leq \tau \right] = \tau = -1$. Finally, if $t_{opt} \neq \tau$, $F_X(t_{opt})$ is smaller than $\rho_{Th}(r)$ and therefore $t_{opt} < \tau$.  
\end{proof}

\subsection{Threshold curve problem} \label{sub3A}

An open question in GM-Th-QAOA is about the behavior of the expectation value as a function of the threshold, fixing optimal $\bm{\beta}$ and $\bm{\gamma}$ for a given choice of $t$. In the present work, we call this function a \textit{threshold curve}. It has been numerically observed that the threshold curve decreases monotonically up to a valley value and then increases monotonically~\cite {GM-Th-QAOA}---recall that we are considering minimization problems. In the numerical experiments of Sec.~\ref{sec4}, we show in Fig.~\ref{curves} an example of that behavior using quantities for abscissa and ordinate in which the threshold curve is equivalent. If that property holds in general, a modified binary search is applicable to find the optimal threshold, representing an exponential improvement over the required linear search otherwise.

Since Theorem~\ref{thm1} gives a closed-form expression for the expectation value, we can directly tackle the problem of the threshold curve. To include the possibility of the threshold curve being constant for a consecutive pair of points, the considered behaviors in our proof are non-increasing and non-decreasing monotonicity instead of strictly decreasing and strictly increasing, respectively. We proved that the threshold curve must change its monotonicity only one time by establishing the derivative change of the sign one time. Using the step function form of $F_Y(T)$ and $G_Y(T)$, we can extend results about monotonicity for the original discrete random variable since it preserves the monotone behavior between any pair of consecutive points of the support. That was done in Theorem~\ref{thm2}, proved in Appendix~\ref{ap1}.

\begin{theorem} \label{thm2}
For any number of layers in GM-Th-QAOA, the threshold curve is monotonically non-increasing up to a valley value and monotonically non-decreasing from there.
\end{theorem}

\subsection{Asymptotic tight bound on quantile} \label{sub3B}

An alternative metric on the performance of GM-Th-QAOA over the expectation value and approximation ratio is the quantity $F_X(E_r(t))$, which corresponds to the quantile in which the expectation value of GM-Th-QAOA is associated. That metric has as a strong point the possibility of comparing the obtained result with the spectrum of the distribution itself, which allows a fair comparison between different distributions and optimization problems. The immediate upper bound $F_X(E_r(t)_{opt}) \leq F_X(\tau)$ can be obtained by applying the cdf to both sides of the inequality in Corollary~\ref{crlr1}. If we assume a continuous distribution, there is a $t$ in which $F_X(t) = \rho_{Th}(r)$ for all $r$ and then $F_X(\tau) = \rho_{Th}(r)$ for any $r$. That way,
\begin{equation}
    \label{eqCb01}
F_X(E_r(t)_{opt})  \leq \sin^2{\left( \frac{\pi}{4r + 2} \right)} = \mathcal{O} \left(\frac{1}{r^2}\right).
\end{equation}
The assumption $X$ as a continuous random variable is convenient since discussing quantiles is more naturally suited for such distributions. We demonstrate in Theorem~\ref{thm3} that the asymptotic bound of Eq.~\eqref{eqCb01} is tight. To do so, we rely on the supposition that $R_X^{min}$ has a finite and non-zero value in pdf. That assumption is also quite reasonable since all target problems of QAOA have a finite optimal value, and the limits $f_X(R_X^{min}) \rightarrow 0$ or $f_X(R_X^{min}) \rightarrow \infty$ are just convenient mathematical abstractions in some situations.

\begin{theorem} \label{thm3}
For GM-Th-QAOA, if $X$ is a continuous distribution and $f_X(R_X^{min}) = a$, where $0 < a < \infty$, then the quantile achieved by the optimal expectation value is asymptotically given by
\begin{equation}
    \label{eqCb02}
F_X(E_r(t)_{opt}) = \Theta \left(\frac{1}{r^2}\right).
\end{equation}
\end{theorem}

\begin{proof}
The upper bound has already been established. For the lower bound, let $t$ be a fixed optimal threshold. We claim that $F_X(t) = F_Y(T) = \Theta(1/r^2)$. If $F_X(t) \notin \Omega(1/r^2)$, then
\begin{equation}
    \label{eqCb03}
    P(\rho, r) = \sin^2{((2r + 1)\arcsin{(\sqrt{\rho}})))} \rightarrow 0,
\end{equation}
on large $r$
\begin{equation}
    \label{eqCb04}
    E_r(t)_{opt} \rightarrow \mu + \operatorname{E}[Y| Y \leq T] P(\rho, r),
\end{equation}
and as $\operatorname{E}[Y| Y \leq T]$ must be bounded by assumption, $E_r(t)_{opt} \rightarrow \mu$, which of course is not an optimal threshold.

Now, consider the bound
\begin{equation}
    \label{eqCb05}
    \begin{split}
    E_r(t)_{opt} \geq \operatorname{E}[X| X \leq t],
    \end{split}
\end{equation} 
which follows from $G_X(t) \leq 0$ and $P(\rho, r) \leq 1$ on Eq.~\eqref{eqC10}. As $F_X(t) = \Theta(1/r^2)$, we just need to prove that the limit
\begin{equation}
    \label{eqCb06}
   L =  \lim_{t \rightarrow R_X^{min}} \frac{F_X(\operatorname{E}[X|X \leq t])}{F_X(t)} 
\end{equation}
is always non-zero finite. Note that the fraction on the limit calculates the cdf of the expected value of distribution $X$ given $X \leq t$. The limit is an indeterminate of $0/0$ type. Applying L'Hôpital's rule using derivatives $F'_X(t) = f_X(t)$ and $G'_X(t) = t f_X(t)$ gives
\begin{equation}
    \label{eqCb07}
L =  a \ {\lim_{t \rightarrow R_X^{min}}} \frac{t F_X(t) - G_X(t)}{F_X(t)^2},
\end{equation}
which is another $0/0$ indeterminate. Therefore,
\begin{equation}
    \label{eqCb08}
L = a \ \lim_{t \rightarrow R_X^{min}} \frac{t f_X(t) + F_X(t) - t f_X(t)}{2 F_X(t) f_X(t)} = \frac{1}{2},
\end{equation}
as desired.

\end{proof}

The theorem establishes a tight quadratic Grover-style speed-up of GM-Th-QAOA over classical brute force in the asymptotic limit, as it takes $r$ rounds to attain an expectation value at a quantile of order $1/r^2$. The result is expected since the optimal angles of GM-Th-QAOA reduce it to an execution of Grover's algorithm.

\subsection{Upper bounds on the standard score} \label{sub3C}

By introducing the auxiliary random variable $Y$, we neutralize the impact of the mean of $X$ with the trivial term $\mu$ in the expression of expectation value. We can do an analog procedure for the standard deviation using the random variable $Z$. With the properties $F_Y(T) = F_Z(T/\sigma) $ and $G_Y(T) = \sigma G_Z(T/\sigma)$ we immediately prove Corollary~\ref{crlr2}, another corollary of Theorem~\ref{thm1}.

\begin{corollary} \label{crlr2}
For any number $r$ of layers in GM-Th-QAOA with optimal angles, the expectation value is given by 
\begin{equation}
    \label{eqCc01}
E_r(t) = \mu - \sigma G_Z(T/\sigma) \frac{1 - P(\rho, r)/F_Z(T/\sigma) } {1 - F_Z(T/\sigma)},
\end{equation}
where $\rho = F_Z(T/\sigma)$.
\end{corollary}

The corollary implies that $E_r(t)$ deviates from the mean proportionally to $\sigma$. Therefore, for a given $X$, it is straightforward to consider the negative of the standard score, i.e., a $C_r(t)$ such that 
\begin{equation}
    \label{eqCc02}
E_r(t) = \mu - C_r(t) \sigma,
\end{equation}
as another performance metric. Note that the performance of GM-Th-QAOA, quantified by $C_r(t)$, is invariant over shifting the location of the distribution or changing its scale, being therefore dependent only on the statistical moments of higher order, such as skewness and kurtosis.

We can get the same conclusion about the expectation value of GM-QAOA with the analysis of Headley and Wilhelm~\cite{headley_paper}. From Equations (D15) and (D17) of Headley and Wilhelm~\cite{headley_paper} paper, denoting the characteristic function by 
\begin{equation}
    \label{eqCc002}
\varphi_X(\gamma) = \sum_{x \in R_X}  f_X(x) e^{i\gamma x}
\end{equation}
and setting $B(\beta) = -1 + e^{i\beta}$,
\begin{widetext}
\begin{equation}
    \label{eqCc03}
    \begin{split}
       & E_r(\bm{\beta}, \bm{\gamma}) = - i \sum_{k_\text{bra}, k_\text{ket} = 0}^{2^r - 1}  \prod_{P \in P_{bra}} \varphi_X \left( \sum_{j \in P} \gamma_j \right) \varphi'_X{\left(\sum_{j \in P_\text{central}} \gamma_j\right)}
        \prod_{P \in P_\text{ket}} \varphi_X \left( \sum_{j \in P} \gamma_j \right) \prod_{-j: k_\text{bra}^j = 1}  B(\beta_j)  \prod_{j: k_\text{ket}^j = 1}  B(\beta_j) \\ & = \mu + 2 \operatorname{Im} \left\{ \sum_{k_\text{bra} < k_\text{ket} = 0}^{2^r - 1}  \prod_{P \in P_{bra}} \varphi_Y \left( \sum_{j \in P} \gamma_j \right) \varphi'_Y{\left(\sum_{j \in P_\text{central}} \gamma_j\right)}
        \prod_{P \in P_\text{ket}} \varphi_Y \left( \sum_{j \in P} \gamma_j \right) \prod_{-j: k_\text{bra}^j = 1}  B(\beta_j)  \prod_{j: k_\text{ket}^j = 1}  B(\beta_j) \right\} \\ & = \mu + 2 \sigma \operatorname{Im} \left\{ \sum_{k_\text{bra} < k_\text{ket} = 0}^{2^r - 1}  \prod_{P \in P_{bra}} \varphi_Z \left(\sigma \sum_{j \in P} \gamma_j \right) \varphi'_Z{\left( \sigma \sum_{j \in P_\text{central}} \gamma_j\right)}
        \prod_{P \in P_\text{ket}} \varphi_Z \left(\sigma \sum_{j \in P} \gamma_j \right) \prod_{-j: k_\text{bra}^j = 1}  B(\beta_j)  \prod_{j: k_\text{ket}^j = 1}  B(\beta_j) \right\},
    \end{split}
\end{equation}
where $k_\text{bra}^j$ and $k_\text{ket}^j$ are the $j$th bit of the binary representation of $k_\text{bra}$ and $k_\text{ket}$, respectively, and
\begin{equation}
    \label{eqCc04}
    \begin{split}
      & P_\text{bra} = \{ P_{\text{bra}}^k: k \in [0, \text{weight}(k_\text{bra})] \}, \ P_\text{ket} = \{ P_{\text{ket}}^k: k \in [0, \text{weight}(k_\text{ket})] \}, \\ & P_\text{central} = \{ -j: j > \text{max} \{ S_\text{bra} \} \} + \{ j: j > \text{max} \{ S_\text{ket} \} \}, \\
    & S_\text{bra} = \{ 0 \} + \{k: k_\text{bra}^k = 1 \}, \ S_\text{ket} = \{ 0 \} + \{k: k_\text{ket}^k = 1 \},
     \\ & P_\text{bra}^k = \{-j: j > k_\text{bra}^k, j \leq k_\text{bra}^{k + 1} \}, \ P_\text{ket}^k = \{j: j > k_\text{ket}^k, j \leq k_\text{ket}^{k + 1} \},
      \end{split}
\end{equation}
with $\text{weight}(k)$ denoting the bit-weight of $k$.
\end{widetext}

The second equality of Eq.~\eqref{eqCc03} follows from the fact that the insertion of $Y$ gives just a global phase change on the final state of GM-QAOA, and the third from the properties $\varphi_Y(\gamma) = \varphi_Z(\sigma \gamma)$ and $\varphi'_Y(\gamma) = \sigma \varphi_Z(\sigma \gamma)$. Therefore the deviation from the mean is also linear on $\sigma$, and we can consider a $C_r(\bm{\beta}, \bm{\gamma})$ such that $E_r(\bm{\beta}, \bm{\gamma}) = \mu - C_r(\bm{\beta}, \bm{\gamma}) \sigma$. Additionally, we conclude that the parameter $\gamma_j$ is inversely proportional to $\sigma$.

A natural question concerns the general upper bounds for $C_r(t)$ and $C_r(\bm{\beta}, \bm{\gamma})$. We denote the maximum possible value for both by respectively $C^{Th}(r)$ and $C^{GM}(r)$. For $r = 1$, we have $C^{GM}(1) \leq 4$ and $C^{Th}(1) \leq 4$. The first bound follows by using the individual bounds $|\varphi_Y(\gamma)| \leq 1$, $|\varphi'_Y(\gamma)| \leq \operatorname{E}[|Y|] \leq \sigma$, $|B(\beta)| \leq 2$, while the second from the inequality
\begin{equation}
    \label{eqCc05}
    \begin{split}
 |G_Y(T)| \leq |G_Y(0)| = 0.5 \operatorname{E}[|Y|] \leq 0.5 \sigma,
    \end{split}
\end{equation}
applied on the polynomial expectation value expression of $r = 1$. The inequality $\operatorname{E}[|Y|] \leq \sigma$ is a consequence of Jensen's inequality. Note that it is unnecessary to check the interval $F_Y(T) > \rho_{Th}(r)$ since by the definition of conditional expectation, setting $F_Y(T) = \rho_{Th}(r)$, included on the other interval, gives the best bound on the range in which $P(\rho, r) = 1$ holds. The bound on GM-QAOA can be refined to $C^{GM}(1) \leq \frac{8 \sqrt{6}}{9} \approx 2.178$ by using calculus arguments based on the inequality of the second derivative of the characteristic function, $\varphi''_Y(\gamma) \leq \sigma^2$, to bound $|\varphi^*_Y(\gamma) \varphi'_Y(\gamma)|$.

For general $r$, applying inequalities on $|\varphi_Y(\gamma)|$, $|\varphi'_Y(\gamma)|$ and $|B(\beta)|$ is insufficient to obtain a satisfactory bound on GM-QAOA since it would grow exponentially. To be more precise, we can get $C^{GM}(r) \leq \frac{9^r - 1}{2} = \Theta(9^r)$ with combinatorial arguments on Eq.~\eqref{eqCc03}. In contrast, for GM-Th-QAOA, we can replace on Eq.~\eqref{eqC05} with $F_Y(T) \leq \rho_{Th}(r)$, in addition to the previous bound on $|G_Y(T)|$, the known result of Grover's algorithm that the maximum ratio $\eta = P(\rho, r)/F_Y(T)$ is hit on the low-convergence regime~\cite{maoa} with $(2r + 1)^2$, to conclude that $C^{Th}(r) = \mathcal{O}(r^2)$. 

Indeed, the above bound is not tight. The tight upper bound is established through the assistance of Lemma~\ref{lm1}, proved in Appendix~\ref{ap2}, which claims that $C^{Th}(r)$ is attained with a particular family of distribution: the two-point distributions. For such distributions, since the random variable $Z$ depends only on the ratio between the probability of the points, we can refer to our defined binary function without loss of generality.

\begin{lemma} \label{lm1}
For any number $r$ of layers in GM-Th-QAOA, the maximum standard score $C_r(t)$ achieved by GM-Th-QAOA, $C^{Th}(r)$, is hit by a two-point distribution.
\end{lemma}

For a given $r$, $C^{Th}(r)$ can be founded by systematically varying the parameter $\rho$ on binary function in the range $0 < \rho < 1$. The choice of threshold trivially is $t = -1$ for binary function, and in that way, $T_h(k)$ is precisely the original binary function. Therefore,
\begin{equation}
\label{eqCc06}
    C_r(t) = \frac{P(\rho, r) + \mu}{\sigma} = \frac{P(\rho, r) - \rho}{\sqrt{\rho(1-\rho)}}.
\end{equation}

\begin{figure}[H]
\centering
\includegraphics[width=1\linewidth]{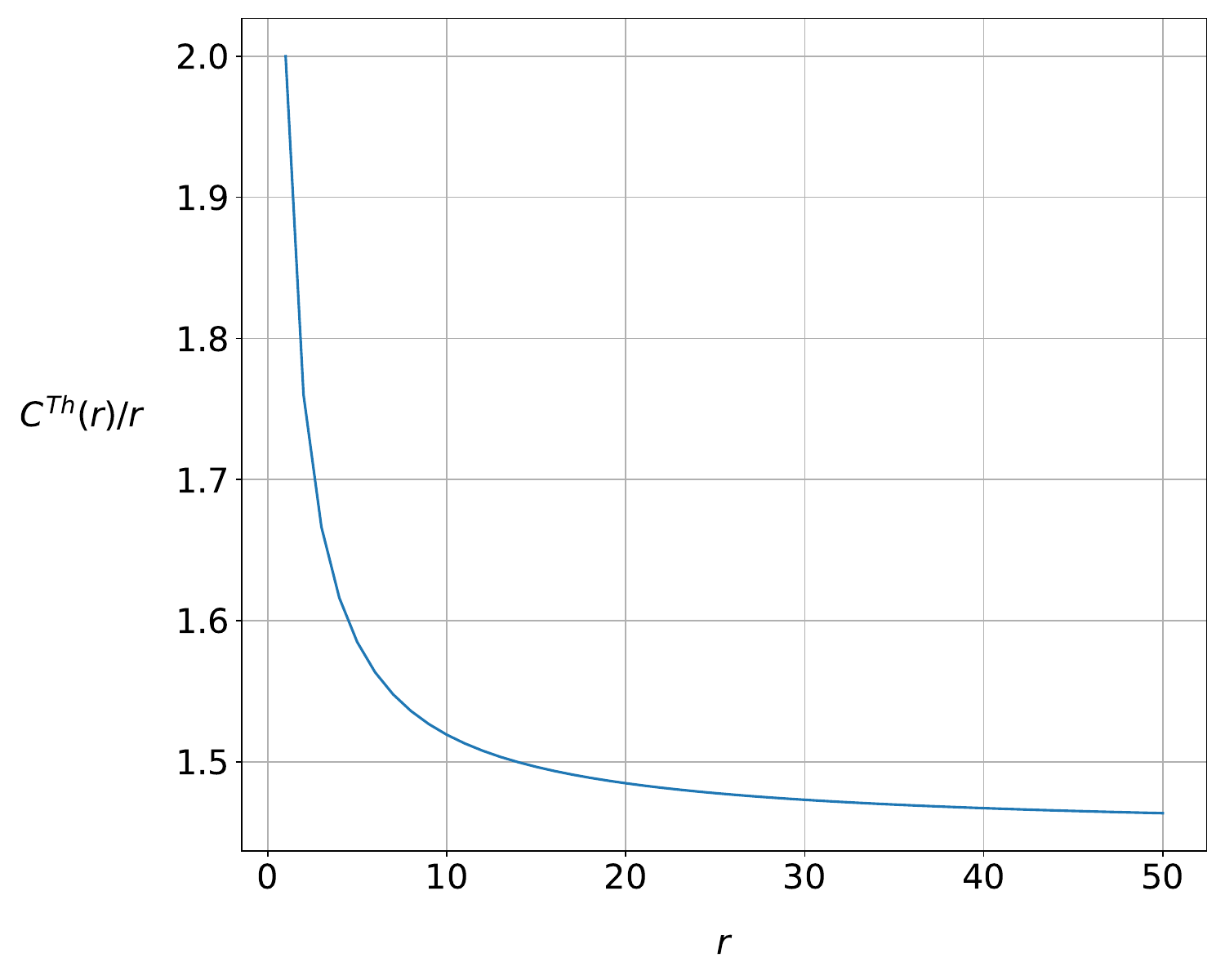}
\caption{The inclination of the curve $C^{Th}(r)$ versus $r$ asymptotically converges to a certain value.}
\label{C_Th_r}
\end{figure}

For $r = 1$, it can be concluded that $C^{Th}(1) = 2$ analytically computing the derivative on the polynomial formula of $P(\rho, 1)$. On the other hand, we solve numerically for $r > 1$. The growth observed is linear in $r$, as shown by Fig.~\ref{C_Th_r}, which plots the ratio $C^{Th}(r)/r$ versus $r$ up to $50$ layers. The inclination of the linear curve converges to a value called $\kappa$. In fact, we prove in Theorem~\ref{thm4} that $C^{Th}(r) = \Theta(r)$ and the value of $\kappa$ is approximately $1.4482$. 

\begin{figure*}[ht]
    \centering
    \subfigure[]{\includegraphics[width=0.49\textwidth]{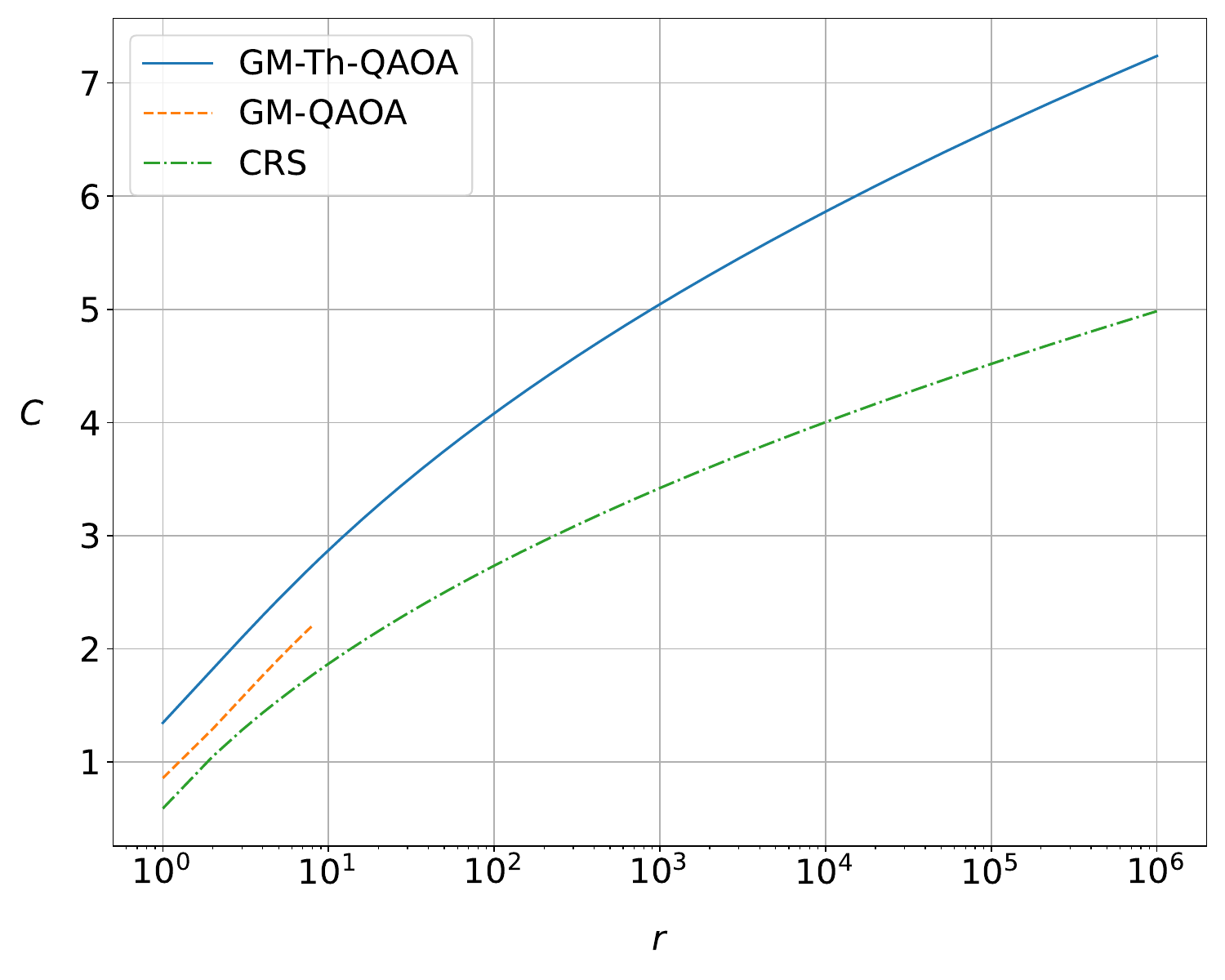}} 
    \subfigure[]{\includegraphics[width=0.49\textwidth]{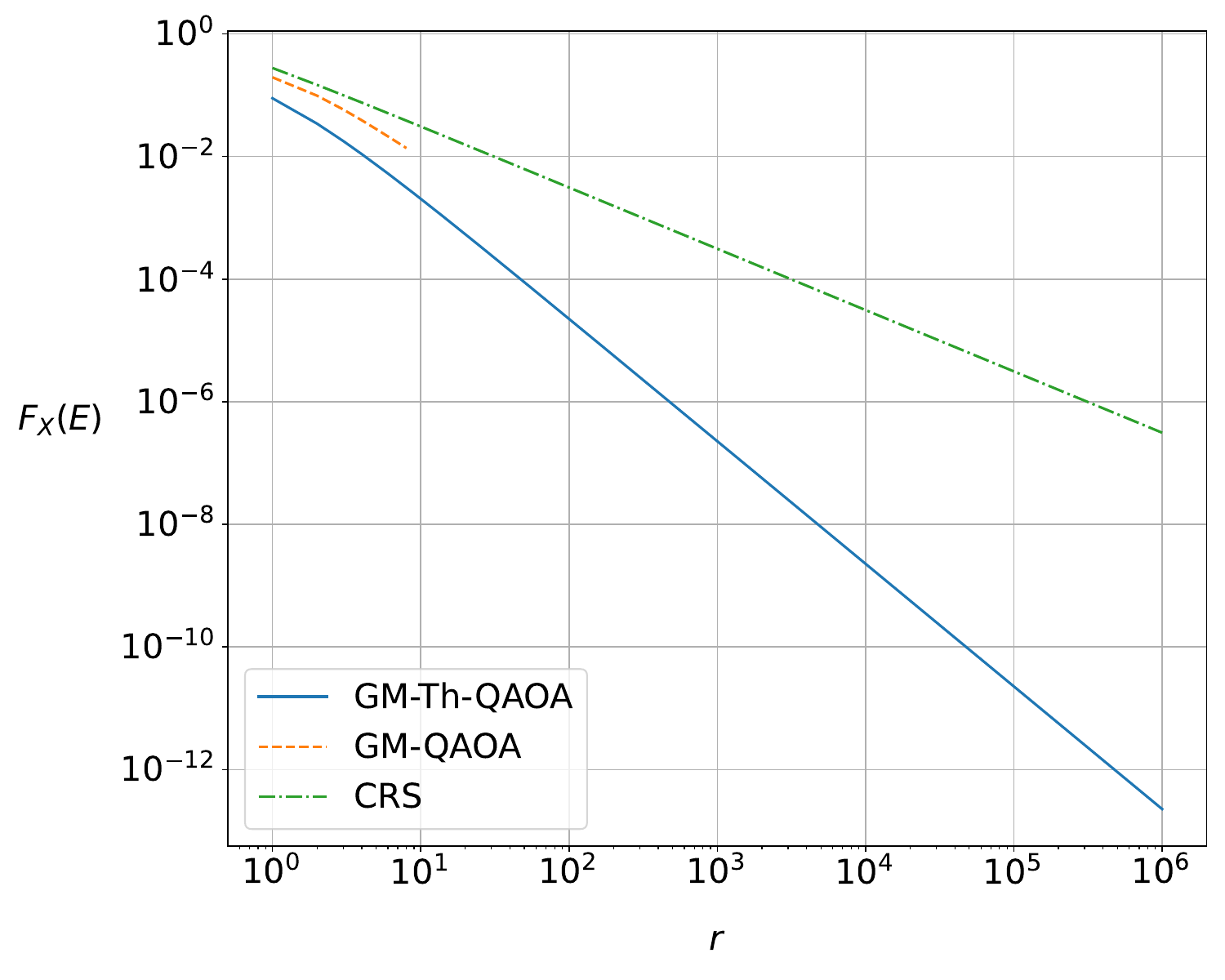}} 
    \caption{Simulation of distribution Normal($u, s^2$) for GM-Th-QAOA and CRS up to $10^{6}$ layers, and for GM-QAOA, due to its exponential complexity, up to $8$ layers. GM-Th-QAOA consistently overcomes GM-QAOA, as expected from the numerical results of the literature, and CRS, consistently with the quadratic gain. (a) Standard score, generically denoted by $C$, versus $r$ in a linear-log scale graphic. The asymptotic behavior of GM-Th-QAOA on $C_r(t)$ indicates a logarithmic growth, according to the expected from the exponential decay of the cdf on $x \rightarrow -\infty$. (b) Log-log graphic of the quantile achieved by the algorithms, generically denoted $F_X(\operatorname{E})$, as a function of $r$. The asymptotic behavior of the cdf illustrates the quadratic gain of GM-Th-QAOA over classical brute force with the quantum algorithm scaling on a $1/r^2$ rate and the classical on $1/r$.}
    \label{normal}
\end{figure*}

\begin{theorem}\label{thm4}
On the large limit of the number of layers $r$, the maximum standard score $C_r(t)$ achieved by GM-Th-QAOA is given by $C^{Th}(r) = \kappa r$, where $\kappa = 2\sin^2{(x_1)}/x_1$ for $x_1$ being the smallest positive solution of the equation $2x = \tan(x)$.
\end{theorem}

\begin{proof}

In the large limit of $r$, Eq.~\eqref{eqCc06} in the interval $\rho \leq \rho_{Th}(r)$ becomes $C_r(t) =\sin^2{\left(2r\sqrt{\rho}\right)}/\sqrt{\rho}$. Taking the derivative equal to $0$ results in $4r \sqrt{\rho} = \tan{\left(2r \sqrt{\rho}\right)}$. The substitution $x = 2r \sqrt{\rho}$ gives the transcendental equation $2x = \tan{(x)}$. The unique positive solution in which $\rho$ is not above $\rho_{Th}(r)$ gives $C^{Th}(r)$.
\end{proof}

Since the binary function is the same as GM-Th-QAOA in GM-QAOA, follows the lower bound $C^{GM}(r) \geq \kappa r$ on large $r$.  In particular, for $r = 1$, $2 \leq C^{GM}(1) \leq \frac{8 \sqrt{6}}{9}$. Furthermore, the upper bound $C^{Th}(r)$ provides an explicit lower bound on the number of round $r$ to reach a fixed approximation ratio $\lambda$, given by Corollary~\ref{crlr3}, that follows from the definitions of $C_r(t)$ and $\lambda$.

\begin{corollary}\label{crlr3}
    For any number $r$ of layers in GM-Th-QAOA, provided that $R_X^{min} \neq 0$ and $|R_X^{min}| < \infty$,
\begin{equation}
    \label{eqCc07}
    r \geq \frac{\mu - \lambda R_X^{min}}{(C^{Th}(r)/r) \sigma}.
\end{equation}
In particular, on the large limit of $r$, $C^{Th}(r)/r = \kappa$.
\end{corollary}

Note that the approximation ratio in our definition can assume negative values if $R_X^{min} < 0$ and the cost function admits positive values. To finish this subsection, we show another bound on the minimum rounds required to achieve an objective. Specifically, we get the minimum number of rounds for the algorithm finding the optimal with probability $1$ (exact optimization). In that case, the optimal threshold must be $t_{opt} = R_X^{min}$ and we must satisfy $F_X(t_{opt}) = f_X(R_X^{min}) \geq \rho_{Th}(r)$. Therefore, 
\begin{equation}
    \label{eqCc08}
    \begin{split}
    r \geq \frac{1}{4} & \left( \frac{\pi}{\arcsin{(\sqrt{f_X(R_X^{min})})}} - 2 \right) \\ & = \Omega\left(1/\sqrt{f_X(R_X^{min})}\right),
    \end{split}
\end{equation}
as $f_X(R_X^{min}) \rightarrow 0$, a quadratic Grover-like speed-up.

\begin{figure*}[ht]
    \centering
    \subfigure[]{\includegraphics[width=0.49\textwidth]{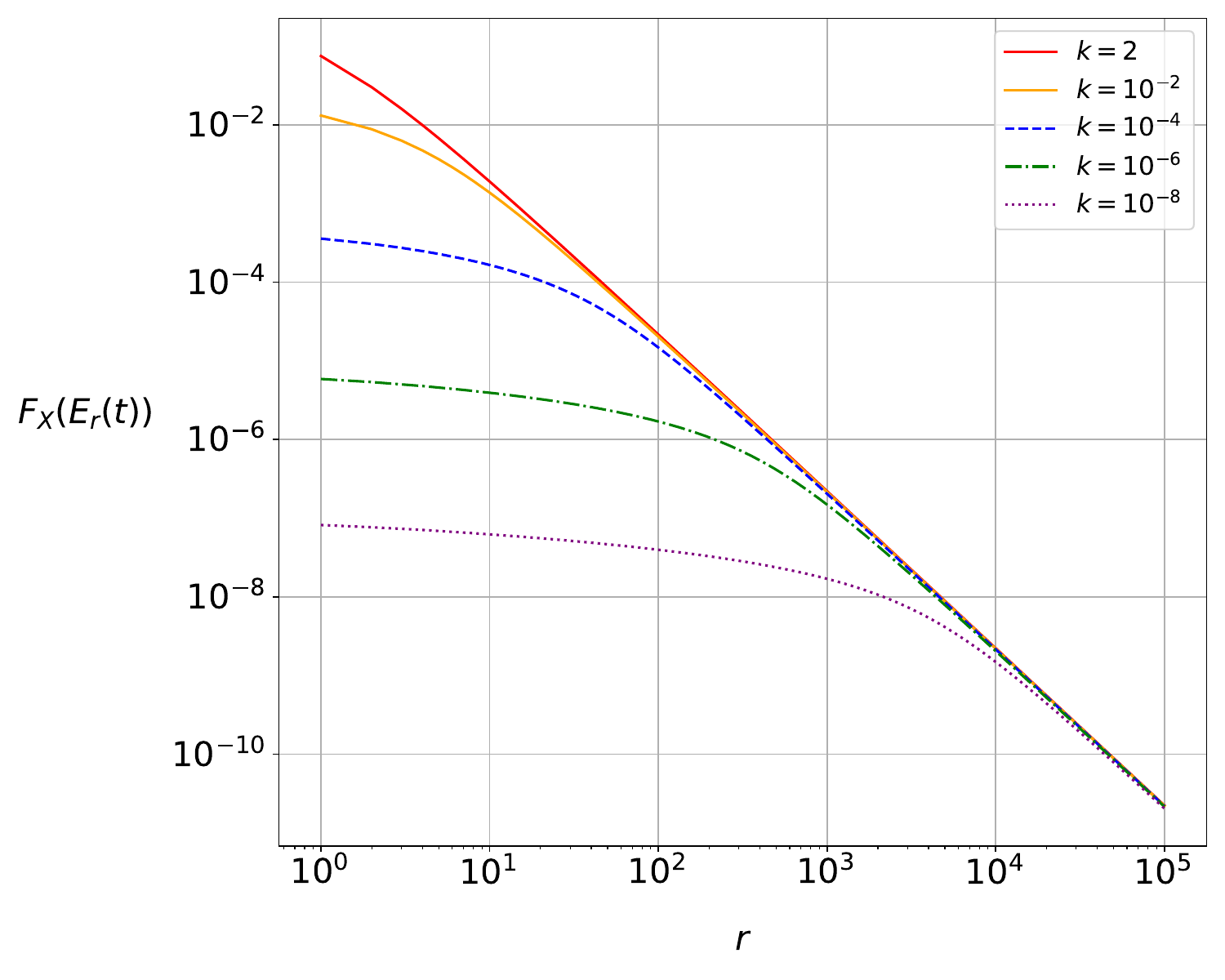}} 
    \subfigure[]{\includegraphics[width=0.49\textwidth]{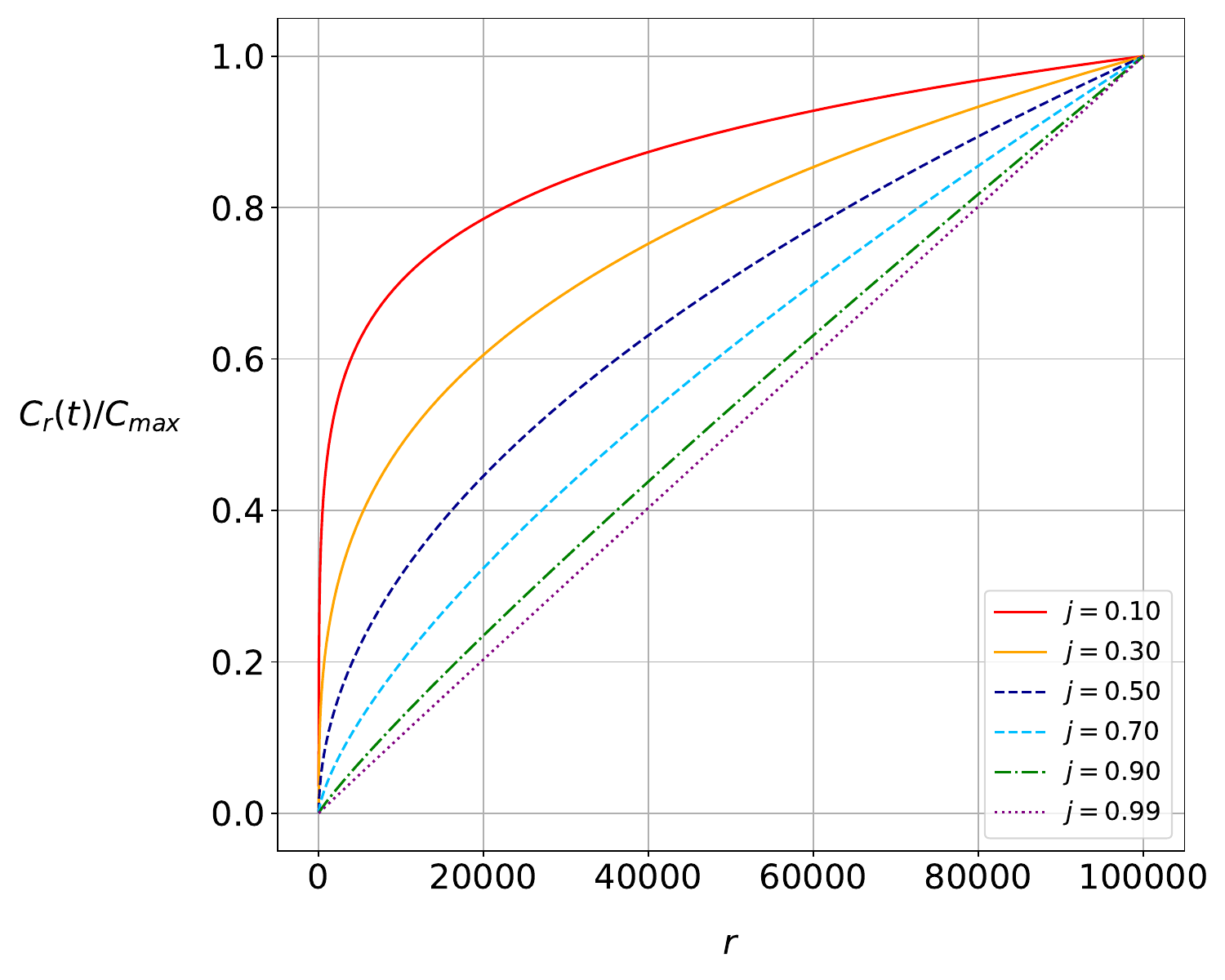}} 
    \caption{(a) Log-log graphic of $F_X(E_r(t))$ versus $r$ up to $10^5$ rounds for distribution Gamma($a, b$) with $a = k/2$ and $b = 1/2$ for values of $k$ decreasing with powers of $10$ (in particular, if $k$ were a positive integer, we would have the chi-squared distribution of $k$ degrees of freedom; and $k = 2$ gives an instance of the reflected exponential distribution). Despite the left-skew of low values of $k$, given a sufficient number of rounds, the asymptotic scale of $1/r^2$ arises, since the cdf of the expected value of the distribution $X$ given $X \leq t$ approaches the limit $L$. (b) Standard score achieved by GM-Th-QAOA up to $10^{5}$ rounds for Pareto($\epsilon, x_m$) with different values of $j$. For viewing purposes, we normalize $C_r(t)$ by $C_{max}$, where $C_{max}$ is the value of $C_{10^5}(t)$. Fitting all curves with a power-law, we found the exponents $0.99, 0.9, 0.7703, 0.5023, 0.3136, 0.1570$ for the respective values of $j$ in descending order. Although the behavior is more precise with the theoretical results on higher values of $j$, the confluence is just a matter of simulating sufficient numbers of layers. For instance, for $j = 0.1$, fitting on $r = 1$ up to $r = x$ for the range $x = 10, 10^2, 10^3, 10^4, 10^5$ gives the progressive improvement of respectively $0.5087, 0.3222, 0.2301, 0.1834, 0.1570$ on the coefficients.}
    \label{gamma_pareto}
\end{figure*}

\subsection{Combining the bounds on the standard score and quantile}
\label{sub3D} 

The explicitly tight bound on the standard score was built using different distributions for each $r$. In particular, the ratio $\rho$ of the two-point distribution that hits $C^{Th}(r)$ changes with $r$. One can ask if a particular distribution gives an asymptotic optimal $C_r(t)$ of order $\Theta(r)$. If this were not the case, we would have the possibility of improving the bound of Corollary~\ref{crlr3} for particular distribution on the asymptotic limit of $r$. However, we can get a family of distributions in which $C_r(t)$ scales arbitrarily close to $\Theta(r)$. The technique to obtain it consists of combining the bound of the quantile of Theorem~\ref{thm3} with the standard score $C_r(t)$. 

To analyze the asymptotic behavior of $C_r(t)$ in terms of $r$, we must assume that $R_X^{min} \rightarrow -\infty$. However, since Theorem~\ref{thm3} has the supposition that $f_X(R_X^{min}) = a$, where $0 < a < \infty$, is necessary the reasonable assumption that the limit $L$ of Eq.~\eqref{eqCb06} is non-zero finite. With the assumption on $L$, so that the result of Theorem~\ref{thm3} be applicable, remains to prove that $F_X(t) = \Theta(1/r^2)$ on $R_X^{min} \rightarrow -\infty$ case. 

To get that, note that if $f_X(-x) = \Omega(1/x^3)$, $\sigma \rightarrow \infty$, since $\operatorname{E}[X^2]$ diverges. Therefore, we must have $f_X(-x) = \mathcal{O}(1/x^3)$. By definition of asymptotic notation, $F_X(-x) = \mathcal{O}(1/x^{2})$ and then $|F^{-1}_X(y)| = \mathcal{O}(1/\sqrt{y})$ as $y \rightarrow 0$. In this way, as $F_X(\operatorname{E}[X| X \leq x])$ scales like $F_X(x)$ by the assumption on $L$, then the conditional expectation has growth rate limit by 
\begin{equation}
    \label{eqCd01}
    |\operatorname{E}[X| X \leq x]| = \mathcal{O}\left(1/\sqrt{F_X(x)}\right)
\end{equation}
as $F_X(x) \rightarrow 0$. Let $t$ be a fixed optimal angle for $r$ rounds of GM-Th-QAOA. As $F_X(t) \leq \rho_{Th}(r)$, $P(\rho, r)$ depends on $F_X(t)$ like $\Theta(F_X(t))$ as $F_X(t) \rightarrow 0$. Note that if we decrease $F_X(t)$, $P(\rho, r)$ also decreases at the same time that $|\operatorname{E}[Y| Y \leq T]|$ increases. However, from Eq.~\eqref{eqCb04}, the growth of $|\operatorname{E}[Y| Y \leq T]|$, bounded by Eq.~\eqref{eqCd01}, cannot compensate the decay of $P(\rho, r)$ and then $|E_r(t)_{opt}|$ is maximized assuming the slowest decay of $F_X(t)$. Therefore, $F_X(t) = \Theta(1/r^2)$, as desired.

Now, consider a $\epsilon > 0$ such that $f_Z(-x) = \Theta(1/x^{3 + \epsilon})$. Repeating the previous argument, $|F^{-1}_Z(y)| = \Theta(1/\sqrt[2 + \epsilon]{y})$ as $y \rightarrow 0$. So, for a fixed optimal threshold $t$, $F_Z(-C_r(t)) = F_X(E_r(t)) = \Theta(1/r^2)$ makes us conclude that $C_r(t) = \Theta(r^{2/(2 + \epsilon)})$. An explicit distribution is a reflected version of Pareto distribution, denoted here Pareto($\epsilon, x_m$), where
\begin{equation}
    \label{eqCd02}
    f_X(x) = \frac{(\epsilon + 2) x_m^{\epsilon + 2}}{(-x)^{\epsilon + 3}}, \ x \in (-\infty, -x_m],
\end{equation}
with parameter $\epsilon > 0$. The limit $L$ given by
\begin{equation}
    \label{eqCd03}
    L = {\left(1 + \frac{1}{1 + \epsilon}\right)^{-(2 + \epsilon)}},
\end{equation}
which lies between $0.25$ and $1/e$, with $L = 0.25$ in $\epsilon \rightarrow 0$ and $L = 1/e$ in $\epsilon \rightarrow \infty$.

In general, the optimal $C_r(t)$ asymptotically depends on $r$ as $C_r(t) = \Theta(|F^{-1}_Z(1/r^2)|)$. If a distribution presents a cdf $F_Z(x)$ with exponential decay on $x \rightarrow -\infty$, the growth of $C_r(t)$ must be logarithmic. It is the case of important distributions of literature, such as the normal, Laplace, gamma (reflected), and exponential (reflected) distributions. Therefore, for an optimization problem with a probability distribution that exhibits a tendency of exponential decay, the number of rounds to achieve a fixed approximation ratio must be exponentially larger than the tight bound of Corollary~\ref{crlr3}.

\section{Numerical experiments} \label{sec4}

\begin{figure*}[ht]
    \centering
    \subfigure[]{\includegraphics[width=0.49\textwidth]{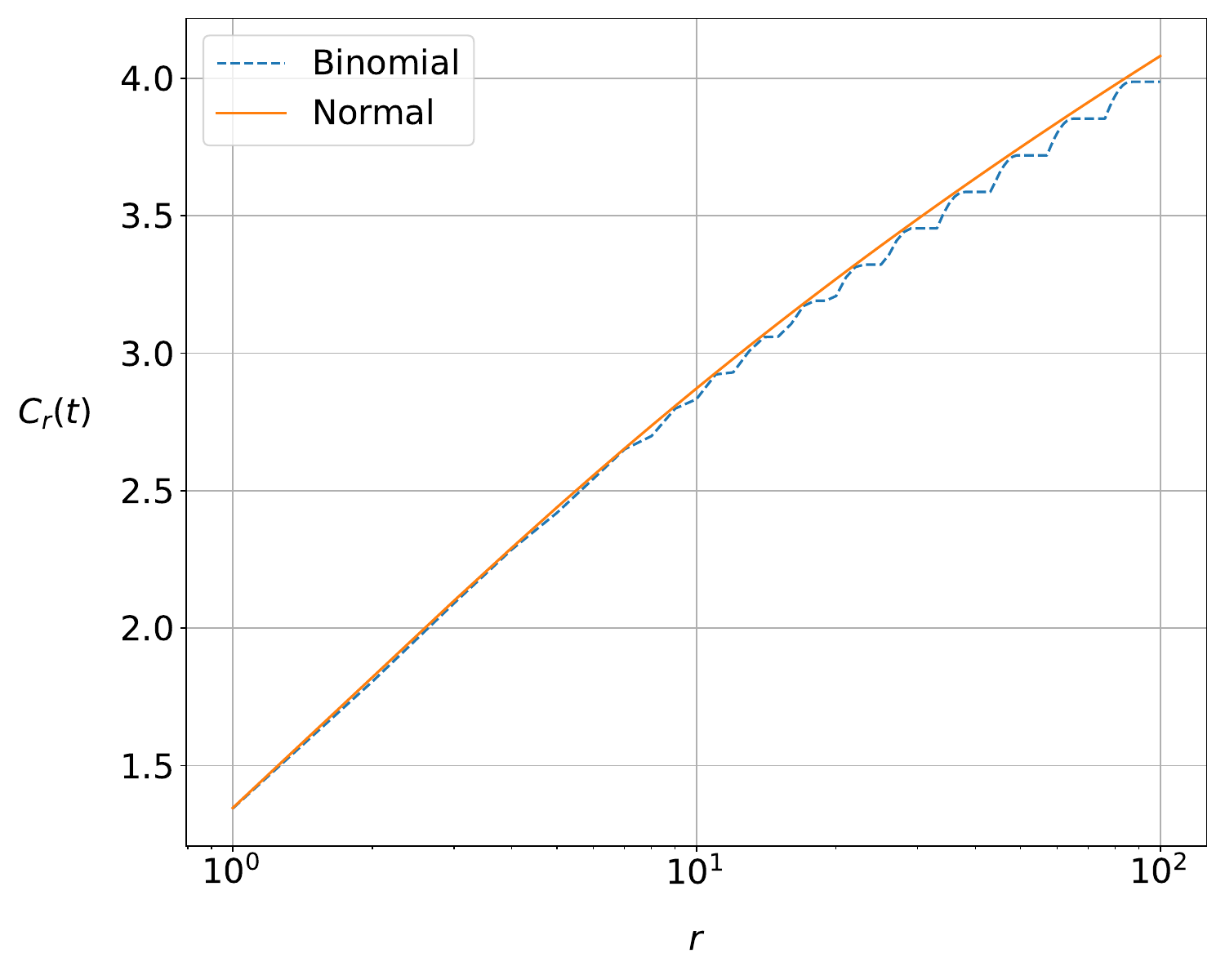}} 
    \subfigure[]{\includegraphics[width=0.49\textwidth]{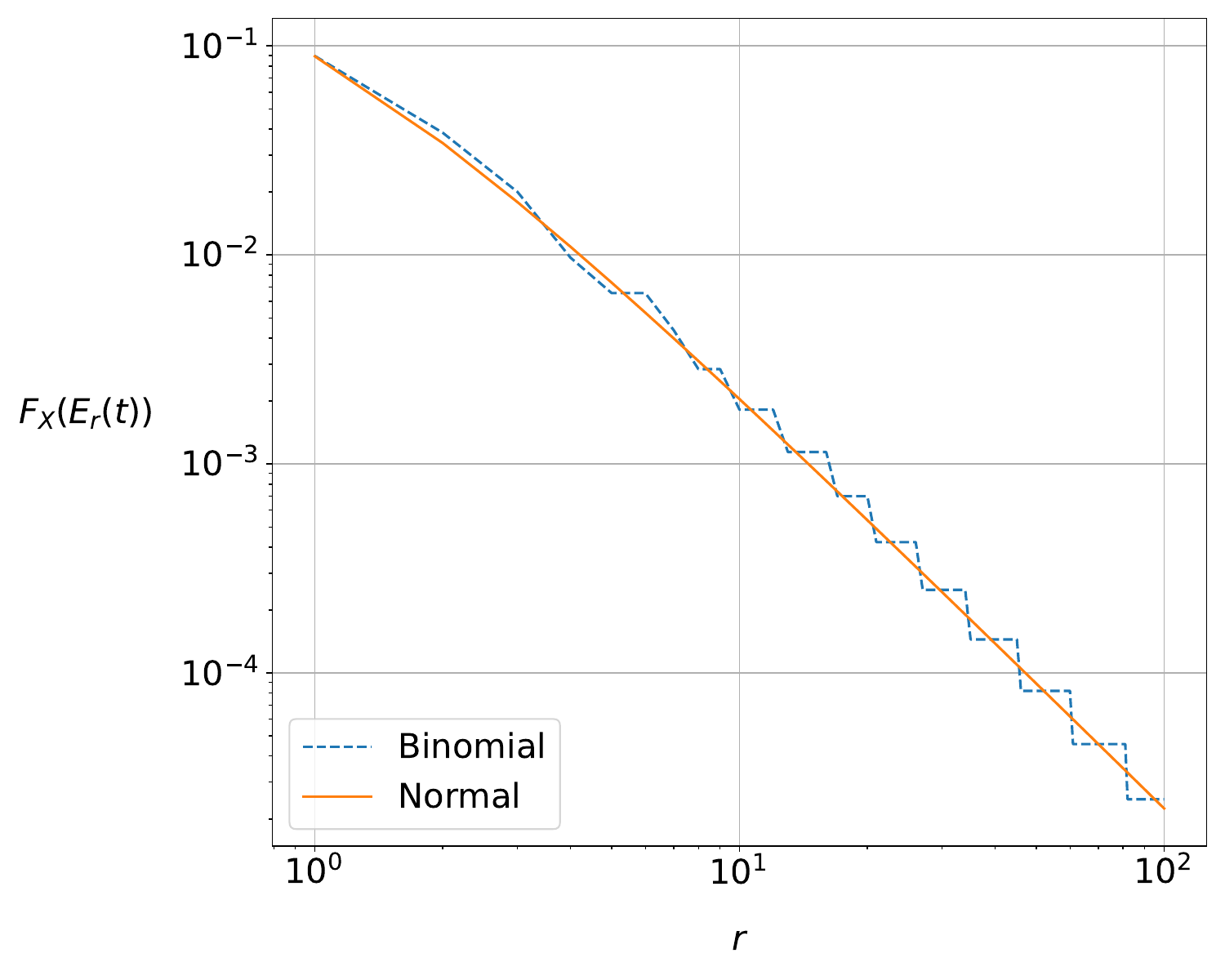}} 
    \caption{Simulation of Binomial($n, p$) with $n = 200$ and $p = 0.5$ for GM-Th-QAOA up to $100$ rounds, compared with the distribution Normal($u, s^2$). (a) $C_r(t)$ versus $r$ on the linear-log scale and (b) $F_X(E_r(t))$ versus $r$ on log-log scale.}
    \label{binomial1}
\end{figure*}

\begin{figure*}[ht]
    \centering
    \subfigure[]{\includegraphics[width=0.49\textwidth]{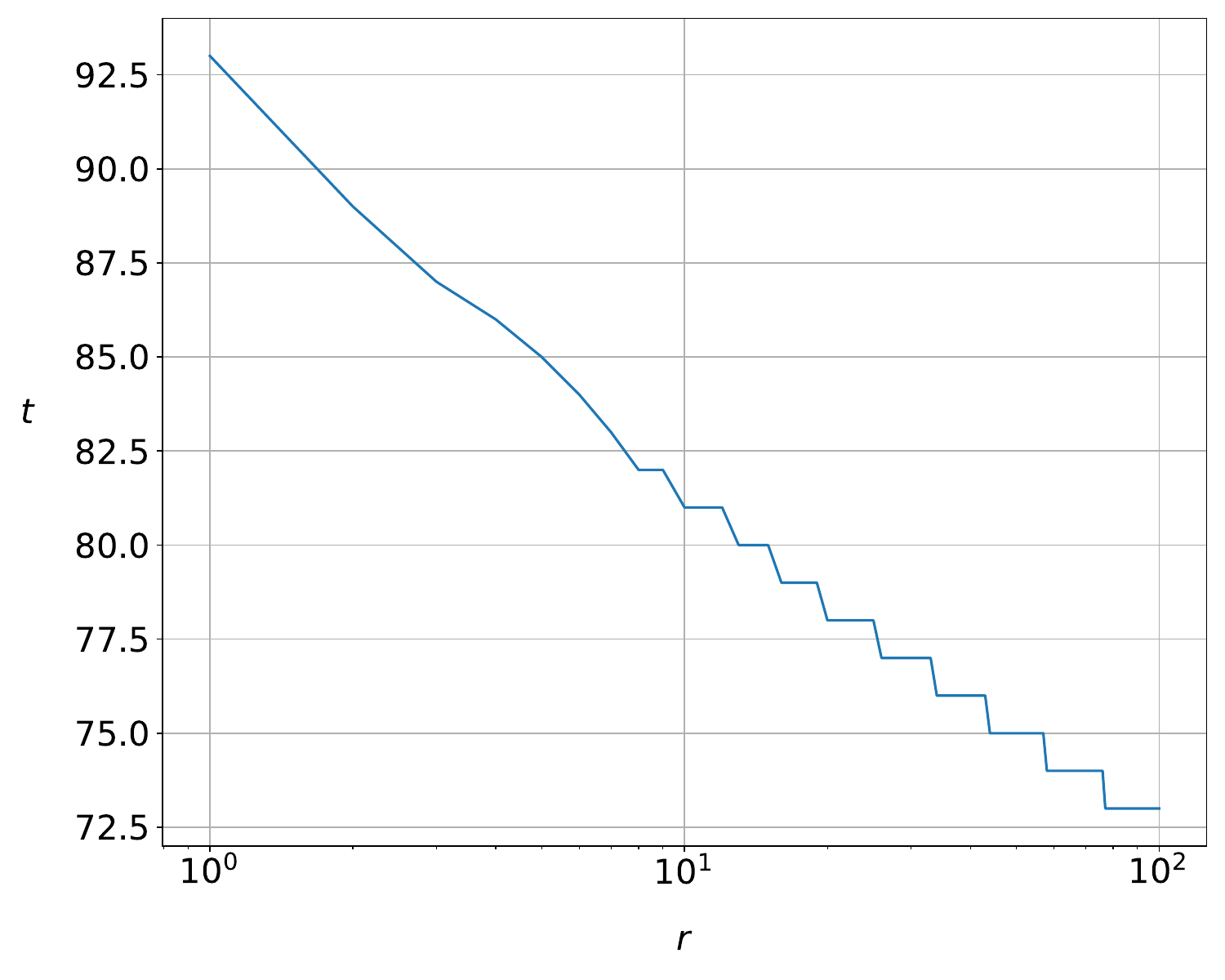}} 
    \subfigure[]{\includegraphics[width=0.49\textwidth]{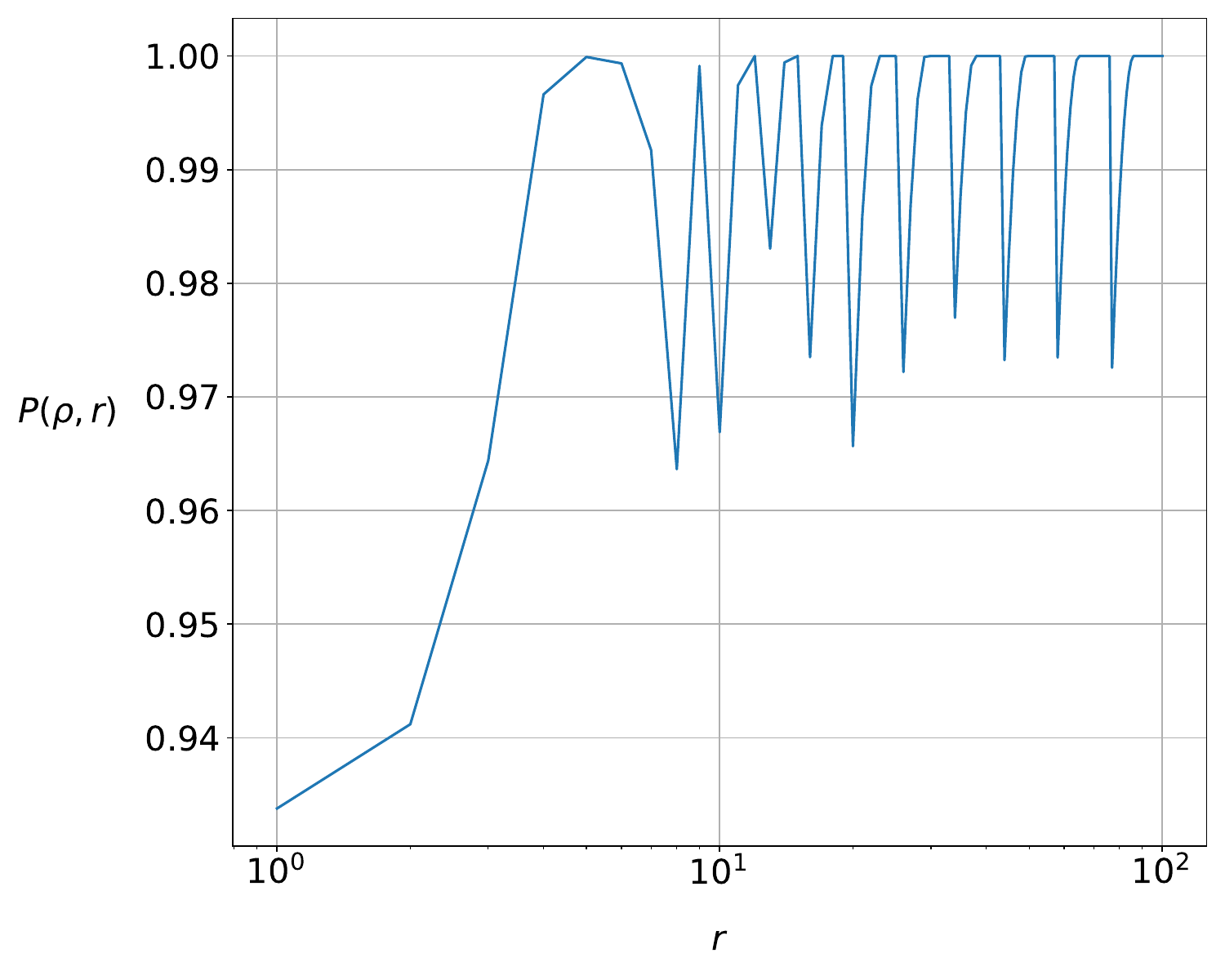}} 
    \caption{(a) The optimal threshold of GM-Th-QAOA and its probability (b) $P(\rho, r)$ versus $r$ in linear-log scale for the distribution Binomial($n, p$) with $n = 200$ and $p = 0.5$ up to $100$ rounds. By plot (a), the threshold value starts to stagnate for some rounds after a certain point. For a given value of optimal threshold $t$, evolving the number of rounds, the probability $P(\rho, r)$ increases until eventually arriving at the maximum value of $1$, as observed in the plot of (b). From there, the only way to improve the performance of GM-Th-QAOA is by changing the threshold to the next value, $t - 1$. However, we may need more than one round for the change to be advantageous, and thus, the algorithm stagnates in that interval. Upon reaching $t - 1$, probability returns to below $1$, and the process repeats, which explains the behavior of Fig.~\ref{binomial1}. Indeed, we can observe that the points with probability $1$ of the plot (b) match the stagnation points of Fig.~\ref{binomial1}.}
    \label{binomial2}
\end{figure*}

We provide numerical experiments computing the formula of Theorem~\ref{thm1} with different probability distributions to emphasize important aspects of our analytical results. The distributions considered are Pareto($\epsilon, x_m$), with probability distribution given by Eq.~\eqref{eqCd02}, and the distributions Normal($u, s^2$), Gamma($a, b$), and Binomial($n, p$), with probability distributions given by 
\begin{equation}
    \label{eqD01}
    \begin{split}
        & f_X(x) = \frac{1}{s \sqrt{2 \pi}} e^{-\frac{1}{2}\left(\frac{x - u}{s}\right)^2} , \ x \in (-\infty, \infty),
        \\ & f_X(x) = \frac{b^a}{\Gamma(a)} (-x)^{a - 1} e^{b x}, \ x \in (-\infty, 0),
        \\ & f_X(x) = \binom{n}{x} p^x (1- p)^{n-x}, \ x \in \{0, 1, \ldots, n\},
    \end{split}
\end{equation}
respectively, where $\Gamma(a)$ is the gamma function. Normal($u, s^2$) and Binomial($n, p$) are the usual normal and binomial distributions, respectively, and Gamma($a, b$) is a reflected version of the usual gamma distribution. Theorem~\ref{thm3} is applicable for the considered continuous distributions since the limit $L$ is $1/e$ for both Normal($u, s^2$) and Gamma($a, b$) (for all used instances on gamma distribution), and given by Eq.~\eqref{eqCd03} for Pareto($\epsilon, x_m$).

\begin{figure*}[ht]
    \centering
    \subfigure[]{\includegraphics[width=0.49\textwidth]{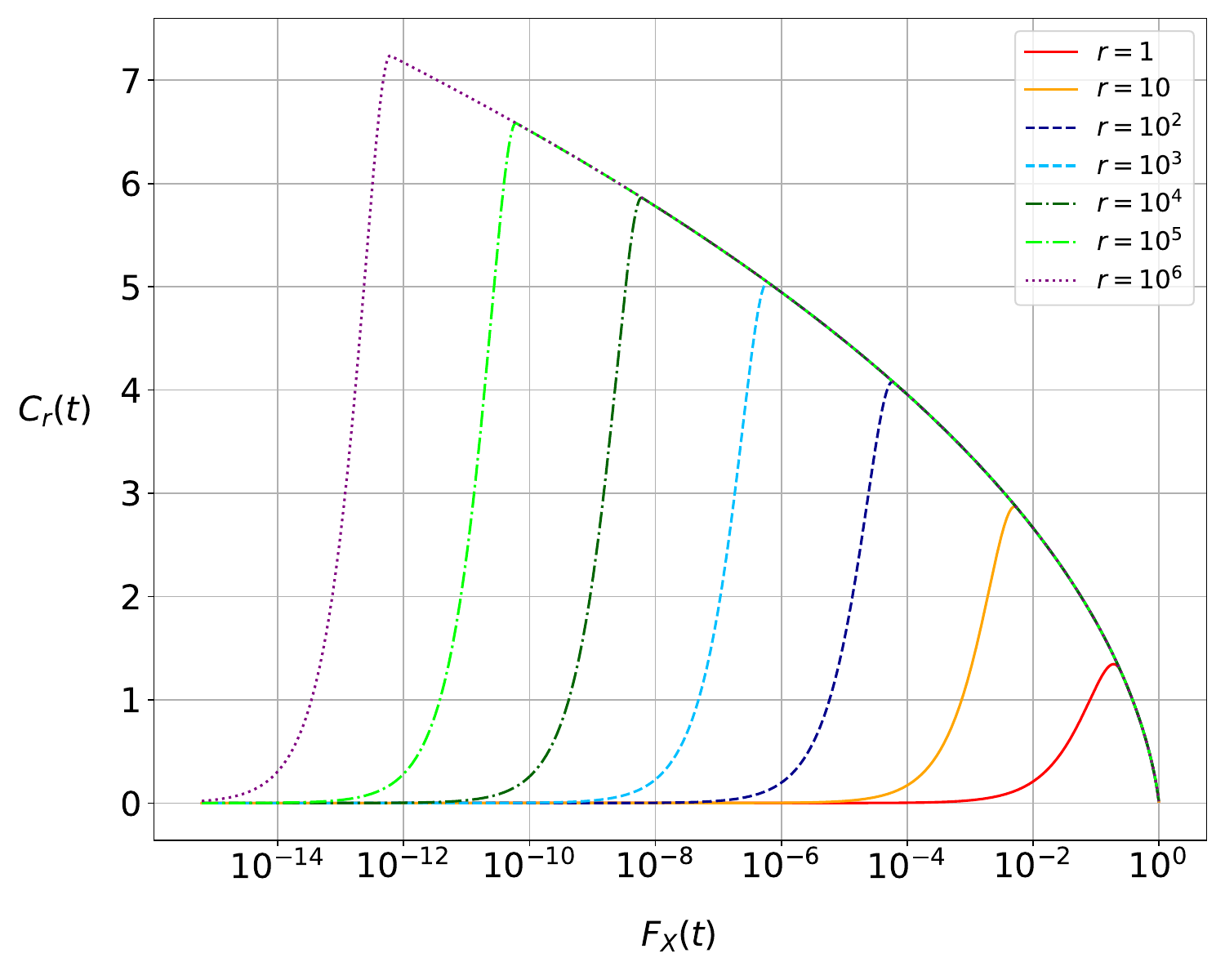}} 
    \subfigure[]{\includegraphics[width=0.49\textwidth]{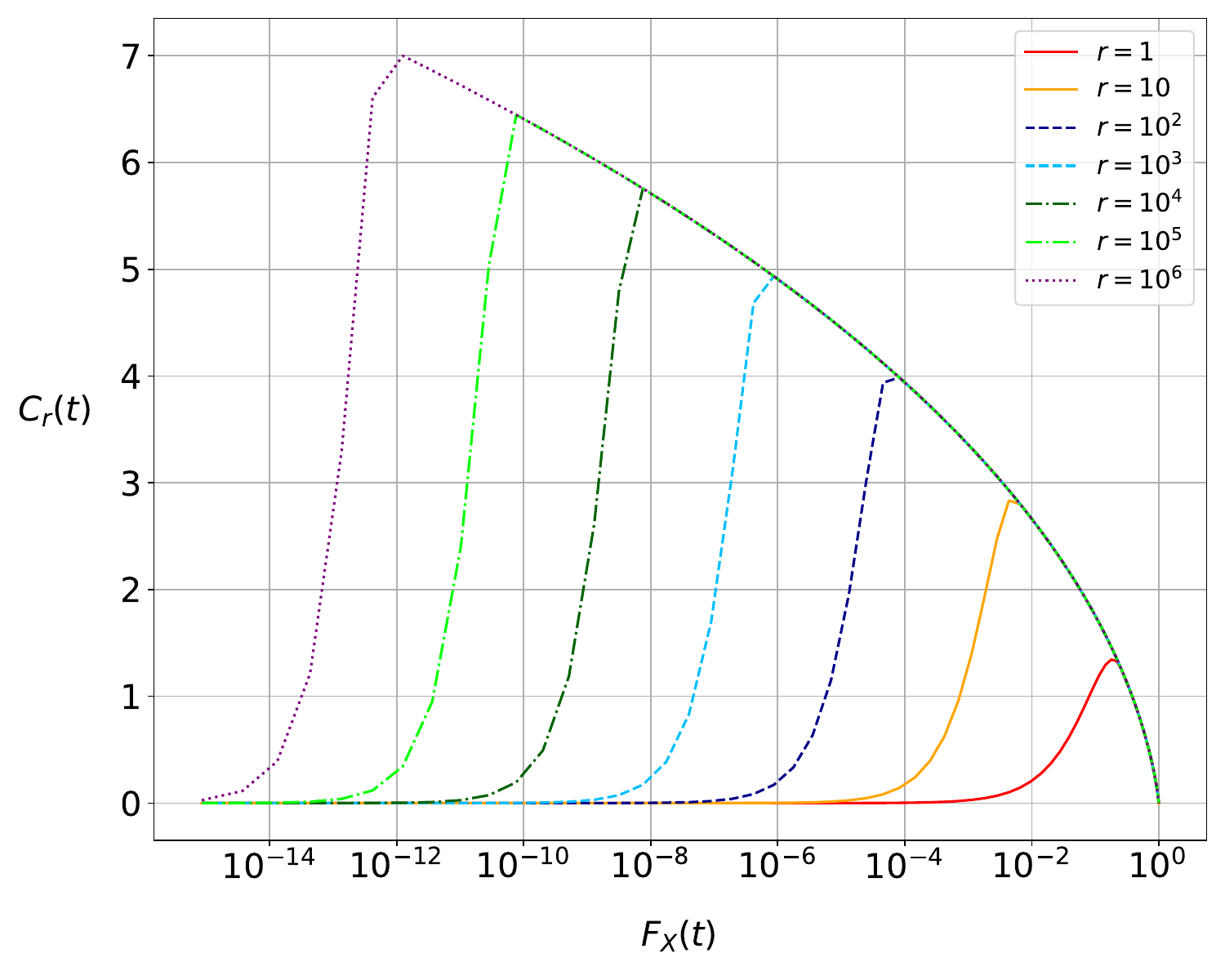}} 
    \caption{Threshold curve of distributions (a) Normal($u, s^2$) and (b) Binomial($n, p$) ($n = 200$ and $p = 0.5$) with $C_r(t)$ versus $F_X(t)$ on a linear-log scale. The resolution considered on the continuous distribution was of $2000$ values for the threshold. For viewing purposes, we show only the values of $r$ in terms of powers of $10$ from $1$ up to $10^{6}$. The envelope that unites the curves is the interval of $P(\rho, r) = 1$.}
    \label{curves}
\end{figure*}

We begin with the normal distribution. The study of that distribution is justified by the ubiquity generated by the central limit theorem. Recall the normally distributed instances of the Capacitated Vehicle Routing and Portfolio Optimization problems~\cite{vehicle, maoa}. It would not be surprising if other optimization problems were normally distributed. The choice of parameters $u$ and $s^2$ on normal distribution is irrelevant to the analysis since the random variable $Z$ is always the standard normal distribution. The results of QAOA are compared with a brute force algorithm known as classical random sampling (CRS)~\cite{vehicle}, which samples a given number of times of the distribution and takes the smallest one. As Bennett et al.~\cite{vehicle}, the equivalent computational effort of CRS to be used is $2r$. The minimum value of a set of samples is known as the first order statistic. The expected value of the first order statistic of the normal distribution is computed by using the Blom~\cite{orderStats} asymptotic approximation,
\begin{equation}
    \label{eqD02}
        u + \Phi^{-1} \left( \frac{1 - c}{2r - 2c + 1} \right) s,
\end{equation}
with $c = 0.375$ and $\Phi$ is the cdf of the standard normal distribution. We simulate GM-Th-QAOA and CRS up to a large number of layers to compare it asymptotically. GM-QAOA, on the other hand, is simulated up to the limit of its exponential complexity. Fig.~\ref{normal} shows how both the standard score and quantile achieved by the algorithms scales, illustrating explicitly the quadratic gain of GM-Th-QAOA over classical brute force, given by Theorem~\ref{thm3}, and the logarithm scale of the standard score, expected from the analysis of Subsec.\ref{sub3D}.

On distribution Gamma($a, b$), we emphasize the asymptotic aspect of Theorem~\ref{thm3}. For that, we simulate the distribution for values of $a$ and $b$ that progressively make it left-skewer. That way, the quantile of the expected value $F_X(\mu)$ decreases so that QAOA already begins with a low quantile and evolves slowly on the first rounds. However, as Fig.~\ref{gamma_pareto}(a) shows, given a sufficient number of rounds, the asymptotic scale of $F_X(E_r(t)) = \Theta(1/r^2)$ appears. For the distribution Pareto($\epsilon, x_m$), we illustrate its asymptotic behavior on $C_r(t)$, obtained on the Subsec.~\ref{sub3D}. To a desired $0 < j < 1$ such that $C_r(t) = \Theta(r^j)$, we can choose the parameter $\epsilon = 2(1-j)/j$. The parameter $x_m$ is a scale parameter and, therefore, irrelevant to the analysis. Fig.~\ref{gamma_pareto}(b) shows the simulation of GM-Th-QAOA for different values of $j$, discussing the asymptotic convergence to the theoretical results.

Finally, we consider the binomial distribution. The distribution Binomial($n, p$) is the sum of $n$ independent Bernoulli random variables with probability $p$ and therefore, by the central limit theorem, approaches normal distribution on $n \rightarrow \infty$. That allows a direct comparison between continuous and discrete distribution to emphasize their differences. Fig.~\ref{binomial1} plots $C_r(t)$ and $F_X(E_r(t))$ versus $r$ for both binomial and normal distributions. As expected from the central limit theorem, both scales similarly. However, note that from a certain $r$ on the binomial distribution, $C_r(t)$ and $F_X(E_r(t))$ do not grow for every increase in $r$, keeping stagnant for some rounds. The case of $F_X(E_r(t))$ can be partially explained by the definition of the cdf on points outside the support $R_X$, but the complete picture is explained in Fig.~\ref{binomial2}, which shows the optimal threshold and its associated probability $P(\rho, r)$, both in a function of $r$. 

Fig.~\ref{curves} shows the threshold curve of binomial and normal distributions for different values of $r$. Without loss of generality, we consider $C_r(t)$ versus $F_X(t)$ instead of the original $E_r(t)$ versus $T$ of the Subsec.~\ref{sub3A}. Both curves are similar and illustrate the result of Theorem~\ref{thm2}, monotonically increasing $C_r(t)$ to up the optimal point and then monotonically decreasing.

\section{General bounds on Grover-based QAOA} \label{sec5}

In Sec.~\ref{sec3}, we prove an asymptotic tight bound that implies GM-Th-QAOA has a quadratic speed-up over the classical brute force approach. That raises the question of whether that bound is general for any Grover-based QAOA, whether GM-QAOA or any potential new variation that can emerge. To answer that question, we develop a technique to get a general upper bound on Grover-based QAOA that consists of getting the maximum amplification of probability over any set of degenerate states. With that upper bound established, we explicitly construct the minimum expectation value within that constrained framework.

The statistical analysis introduced by Headley and Wilhelm~\cite{headley_paper} for GM-QAOA can be applied to Grover-based QAOA. For that, we introduce the subscript on random variables to differ between the distributions of the functions $c(k)$ and $q(k)$. Specifically, for the originals $X$, $Y$, and $Z$ we respectively denoted $X_c$, $Y_c$, and $Z_c$ for $c(k)$, and $X_q$, $Y_q$, and $Z_q$ for $q(k)$. The random variable $X_q$ can be expressed as a mapping from $X_c$ such that $X_q = q(X_c)$. For instance, in the GM-Th-QAOA, we have
\begin{equation}
    \label{eqE01}
    X_q = q(X_c) = 
    \begin{cases}
        -1, &  X_c \leq t\\
        0,              & \text{otherwise}.
    \end{cases}
\end{equation}

Applying the analogous analysis, from Eq.~\eqref{eqCc03}, $\varphi_X(\gamma)$ becomes $\varphi_{X_q}(\gamma)$, while the derivative of the characteristic function is changed from $\varphi'_X(\gamma)$ to $\Psi_X(\gamma)$, where
\begin{equation}
    \label{eqE02}
    \begin{split}
    \Psi_X(\gamma) & = i {\displaystyle \langle s | H_C {U}_{P}(\gamma) | s \rangle} = i \sum_{x \in R_{X_c}} x f_{X_c}(x) e^{i \gamma q(x)}.
    \end{split}
\end{equation}
The symbols $\mu$ and $\sigma$ continue denoting the mean and standard deviation associated with the cost function. In particular, follows the necessary conditions of $\mu = -i \Psi_X(0)$ and $\Psi_Y(\gamma) = \sigma \Psi_Z(\sigma \gamma)$. Additionally, we denote $C_r$, where $E_r = \mu - C_r \sigma$, and  $C(r)$ as the maximum $C_r$ achieved by Grover-based QAOA.

To bound the maximum amplification of the probability over any set of degenerate states, we define $S_T$ as a set of elements on the spectrum of the Hamiltonian $H_Q$ with some fixed cost $x_o$. Suppose that $H_Q$ is built from an arbitrary problem Hamiltonian $H_C$. For a given $r$, our goal is to maximize the ratio between the probability of measuring a state on $S_T$ before and after the application of the QAOA operators optimizing the choices of the ratio $|S_T|/M$ and the probability distribution $f_{X_q}(x)$. The only restriction on the choice of the distribution $f_{X_q}(x)$ is sign the probability $|S_T|/M$ on value $x_o$. To get it, consider taking the expectation value on the final state of Grover-based QAOA of a third Hamiltonian $H_{max}$ that encodes $x_o$ to $|S_T|$ elements and $0$ to the remainders, with ratio $\rho = |S_T|/M$. The probability of measuring an element of $S_T$ on the initial state is $\rho$, while after the application of QAOA operators is $ E_r^{max}(S_T, f_{X_q})/x_o$, where $E_r^{max}(S_T, f_{X_q})$ denotes the expectation value of that configuration. We want to maximize the ratio between then, named $\eta_r(S_T, f_{X_q})$, and use it bound to explicitly build the minimum expectation value on an arbitrary instance of some Grover-based QAOA by sequentially maximally amplifying the states in ascending order of costs until the sum of probabilities reaches $1$. As the amplitudes of degenerate states are equal, the amplification is in ascending order of the support of $X_c$.

For $r = 1$, we get the maximum amplification analytically. The mean and $\Psi_X(\gamma)$ are with respect to $H_{max}$, giving $\mu = x_o \rho$ and
$\Psi_X(\gamma) = i x_o \rho e^{i \gamma x_o}$.
Consequently,
\begin{equation}
    \label{eqE03}
    \begin{split}
\eta_1(S_T, f_{X_q}) & = 1 + | B(\beta)|^2 | \varphi_{X_q}(\gamma)|^2 \\ & + 2 \operatorname{Re} \{e^{i\gamma x_o} B^*(\beta) \varphi^*_{X_q}(\gamma) \}.
    \end{split}
\end{equation}
Since $|B(\beta)| \leq 2$ and $|\varphi_{X_q}(\gamma)| \leq 1$, then $ \eta_1(S_T, f_{X_q}) \leq 9$. That value is saturated if $\rho \rightarrow 0$ and the remainder probability of $f_{X_q}(x)$ is completed on value $0$, i.e., if $f_{X_q}(x)$ represents a binary function up to a scale change of ratio $\rho$, in which $\beta = \pi$ and $\gamma = \pi/x_o$ is optimal.

Note that the maximum amplification is $(2r + 1)^2$, the exact amplification of Grover's algorithm on the low-convergence regime. One can ask if the maximum amplification is on the low-convergence regime for any $r$. There is numerical evidence for that from Bennett and Wang~\cite{maoa} in the context of QWOA on the complete graph. Unfortunately, applying the individual bounds of $|B(\beta)|$ and $|\varphi_{X_q}(\gamma)|$ on general $r$ expression gives an exponential amplification of $9^r$. Moreover, due to the complexity of the expression, direct analytical treatment is unfeasible, necessitating indirect methods. Specifically, we demonstrate in Lemma~\ref{lm2}, proved on Appendix~\ref{ap3}, that the maximum amplification is $(2r + 1)^2$ by showing that the existence of a distribution that can achieve a larger amplification implies an explicit algorithm for the unstructured search problem with a larger average probability than the bound of Hamann, Dunjko, and Wölk~\cite{generalOptimalGrover}. With Lemma~\ref{lm2}, we can prove the lower bound on $E_r$ (i.e., a general upper bound on Grover-based QAOA performance), given by Theorem~\ref{thm5}.

\begin{lemma} \label{lm2}
For any number $r$ of layers in Grover-based QAOA with a set $S_T$ of ratio $\rho$, the amplification of the probability of measuring the elements of $S_T$ is bounded by
\begin{equation}
    \label{eqE04}
\eta_r(S_T, f_{X_q}) \leq (2r + 1)^2,
\end{equation}
where the tight bound is achieved considering the limit of $\rho \rightarrow 0$ and $q(k)$ equal to the binary function up to a scale change.
\end{lemma}

\begin{theorem} \label{thm5}
For any number $r$ of layers in Grover-based QAOA, the expectation value is bounded by
\begin{equation}
    \label{eqE05}
E_r \geq G_{X_c}(\tau_1)(2r +1)^2  + \tau_2 (1 - F_{X_c}(\tau_1)(2r + 1)^2),
\end{equation}
where $\tau_1$ is the maximum element of the support of $X_c$ in which $F_{X_c}(t) \leq 1/(2r + 1)^2$ and $\tau_2$ is the minimum element in which $F_{X_c}(t) > 1/(2r + 1)^2$. In particular, if $F_{X_c}(\tau_1) = 1/(2r + 1)^2$, then $E_r \geq \operatorname{E}[{X_c}| {X_c} \leq \tau_1]$.
\end{theorem}

\begin{proof}
To build our upper bound on expectation value, we assume the largest amplification of $(2r + 1)^2$, bounded by Lemma~\ref{lm2}, for the smallest solutions until $\tau_1$. The remainder probability is assigned to $\tau_2$, resulting in Eq.~\eqref{eqE05}. If $F_{X_c}(\tau_1) = 1/(2r + 1)^2$, the second term vanishes and $E_r \geq \operatorname{E}[{X_c}| {X_c} \leq \tau_1]$.
\end{proof}

The equality of Eq.~\eqref{eqE05} is referred to as the \textit{maximum amplification bound}. The bound is not tight since we can reach probability $1$ on the search problem only if $\rho$ is at least the larger ratio of $\rho_{Th}(r)$, a consequence of the fact that the amplification decreases as we move away from the low-convergence regime. Note that the MAOA operates close to the regime of the maximum amplification, although it does not use the expectation value as a metric. Despite this, the maximum amplification has the same asymptotic behavior as GM-Th-QAOA in all aspects considered---as a result, the same asymptotic behavior emphasized in the numerical experiments of Sec.~\ref{sec4} could be reached by computing the maximum amplification bound. Firstly, if $X$ is continuous, $F_{X_c}(\tau_1) = 1/(2r + 1)^2$ for any $r$ and the bound $E_r \geq \operatorname{E}[{X_c}| {X_c} \leq \tau_1]$ combined with $F_{X_c}(\tau_1) = \Theta(1/r^2)$ gives Corollary~\ref{crlr4}, a generalization of Theorem~\ref{thm3} which follows using analogous arguments. 

\begin{corollary} \label{crlr4}
For Grover-based QAOA, if $X_c$ is a continuous distribution and $f_{X_c}(R_{X_{c}}^{min}) = a$, where $0 < a < \infty$, then the quantile achieved by the expectation value is asymptotically bounded by
\begin{equation}
    \label{eqE06}
F_{X_c}(E_r) = \Omega \left(\frac{1}{r^2}\right). 
\end{equation}
\end{corollary}

Corollary~\ref{crlr4} implies that any Grover-based QAOA cannot be asymptotic better than the quadratic Grover-like speed-up, the most important conclusion of this work. Moreover, all the constructions of Subsec.~\ref{sub3D} are applicable to the maximum amplification bound. Now, combining Corollary~\ref{crlr1} and Theorem~\ref{thm5}, and assuming $X$ continuous gives a comparison of GM-Th-QAOA with maximum amplification bound on the large limit of $r$ of
\begin{equation}
    \label{eqE07}
     \frac{L}{4r^2} \leq F_{X_c}(E_r(t)) \leq \frac{L\pi^2}{16r^2},
\end{equation}
and thus GM-Th-QAOA is, in the worst case, $\pi^2/4$ times worse than the maximum amplification bound in terms of the cdf. With the maximum amplification bound, we can also bound $C(r)$---and the number of rounds to achieve a fixed approximation ratio---obtaining the analogous of Theorem~\ref{thm4} and Corollary~\ref{crlr3} for Grover-based QAOA, synthesized in Theorem~\ref{thm6}, that is proved in Appendix~\ref{ap4}. 

\begin{theorem}\label{thm6}
For any number $r$ of layers in Grover-based QAOA, $C(r) \leq 2 \sqrt{r(r + 1)}$ and, provided that $R_X^{min} \neq 0$ and $|R_X^{min}| < \infty$,
\begin{equation}
    \label{eqE08}
    r \geq \frac{\mu - \lambda R_{X_c}^{min}}{2\sigma \sqrt{1 + 1/r}}.
\end{equation}
\end{theorem}

The bound on the quantity $C(r)/r$ is decreasing in $r$, with a maximum of $2\sqrt{2}$ in $r = 1$ and a minimum of $2$ in $r \rightarrow \infty$. Combining Theorems~\ref{thm4} and~\ref{thm6}, $\kappa r \leq C^{GM}(r) \leq 2r$ on the limit of large $r$. Note that Theorem~\ref{thm6} improve the bound of Theorem 3 on Benchasattabuse et al.~\cite{lowerBounds} paper by a constant factor of $\sqrt{2}\pi$ on $r = 1$ and $2\pi$ on $r \rightarrow \infty$. Beyond the more general context of Grover-based QAOA, our lower bound has the advantage of allowing any cost function. We get also the analogous of the Eq.~\eqref{eqCc08} for the number of rounds to reach probability $1$ of measuring an optimal solution with 
\begin{equation}
    \label{eqE09}
    r \geq \frac{1}{2}\left( \frac{1}{\sqrt{f_{X_C}(R_{X_c}^{min})}} - 1 \right) = \Omega\left(1/\sqrt{f_{X_c}(R_{X_c}^{min})}\right)
\end{equation}
as $f_{X_c}(R_{X_c}^{min}) \rightarrow 0$

Furthermore, a direct comparison with Grover Adaptive Search follows directly from the bound on amplification of Lemma~\ref{lm2}. Since the probability is bounded by $f_{X_c}(R_{X_c}^{min}) (2r + 1)^2$, finding an optimal solution for an optimization problem with Grover-based QAOA with probability of at least, for instance, $1/2$, needs $\Omega(1/\sqrt{f_{X_c}(R_{X_c}^{min})})$ rounds, an analogous result to the Theorem 1 of Durr and Hoyer~\cite{gas1}.

To conclude the analytical discussion, we emphasize that although the optimization metric considered for the observable throughout this paper is the expectation value, the result of Lemma~\ref{lm2}, i.e., the bound on the amplification of the probability of sampled states, is an indication of quadratic gain that is independent of the particular metric of expectation value and may extend the key conclusion of our paper that the performance of Grover-based QAOA is limited to a quadratic speed-up over classical brute force for other metrics. In particular, an important weakness of expectation value as an optimization metric is that, except for the cases of instances that exhibit concentration of the sampled states to the average value on $M \rightarrow \infty$, such as Ref.~\cite{sherrington-kirkpatrick}, it is possible that we have a distribution of the sampled states in which the probability of measuring the best solutions is not adequately amplified, instead more amplifying the states with lower quality~\cite{CVaR, gibbs, samples, maoa}. Alternative optimization metrics were introduced in the literature to address this issue, including the Conditional Value-at-Risk (CVaR)~\cite{CVaR}, the Gibbs Objective Function~\cite{gibbs}, the probability of sampling states above certain quality~\cite{samples}, and even the probability of measuring optimal solutions~\cite{maoa}. However, the bound of Lemma~\ref{lm2} naturally remedies the discussed problem. In particular, as an example of what we did in the proof of Theorem~\ref{thm5} for the expectation value, we can assume the most favorable distribution of sampled states for combinatorial optimization, i.e., assuming the maximum quadratic amplifying the probability of the best states in ascending order of costs.

\subsection{Bounds on Max-Cut problem} \label{sub5A}

One application of our bounds is the Max-Cut, a problem widely considered in QAOA literature~\cite{qaoa_review}. The Max-Cut problem consists of finding a partition of the vertices of a given graph into two complementary subsets that maximize the number of edges between both subsets. The classical Goemans-Williamson algorithm~\cite{GW} provides the best-known approximation ratio for that problem with $\lambda \approx 0.8786$, which induces a point of comparison for the performance of quantum algorithms that aim to find approximate solutions, such as the QAOA. Under the assumption that the unique games conjecture holds, the approximation Goemans-Williamson gives the best approximation ratio guarantee for any polynomial-time algorithm on Max-Cut~\cite{uniqueGames}. Without that unproven assumption, approximate Max-Cut with an approximation ratio $\lambda = 16/17 \approx 0.9411$ is proved to be NP-hard~\cite{MaxCutNPH1, MaxCutNPH2}. In this sense, a widely known analytical result in QAOA literature is that for the particular class of $3$-regular graphs, QAOA with the original transverse field mixer and single layer guarantees an approximation ratio of $\lambda = 0.6924$~\cite{qaoa}. That result was extended to $r = 2$ and $r = 3$, with approximation ratio guarantees of $\lambda = 0.7559$ and $\lambda = 0.7924$, respectively~\cite{maxcut_ar1}. 

\begin{figure*}[ht]
    \centering
    \subfigure[]{\includegraphics[width=0.49\textwidth]{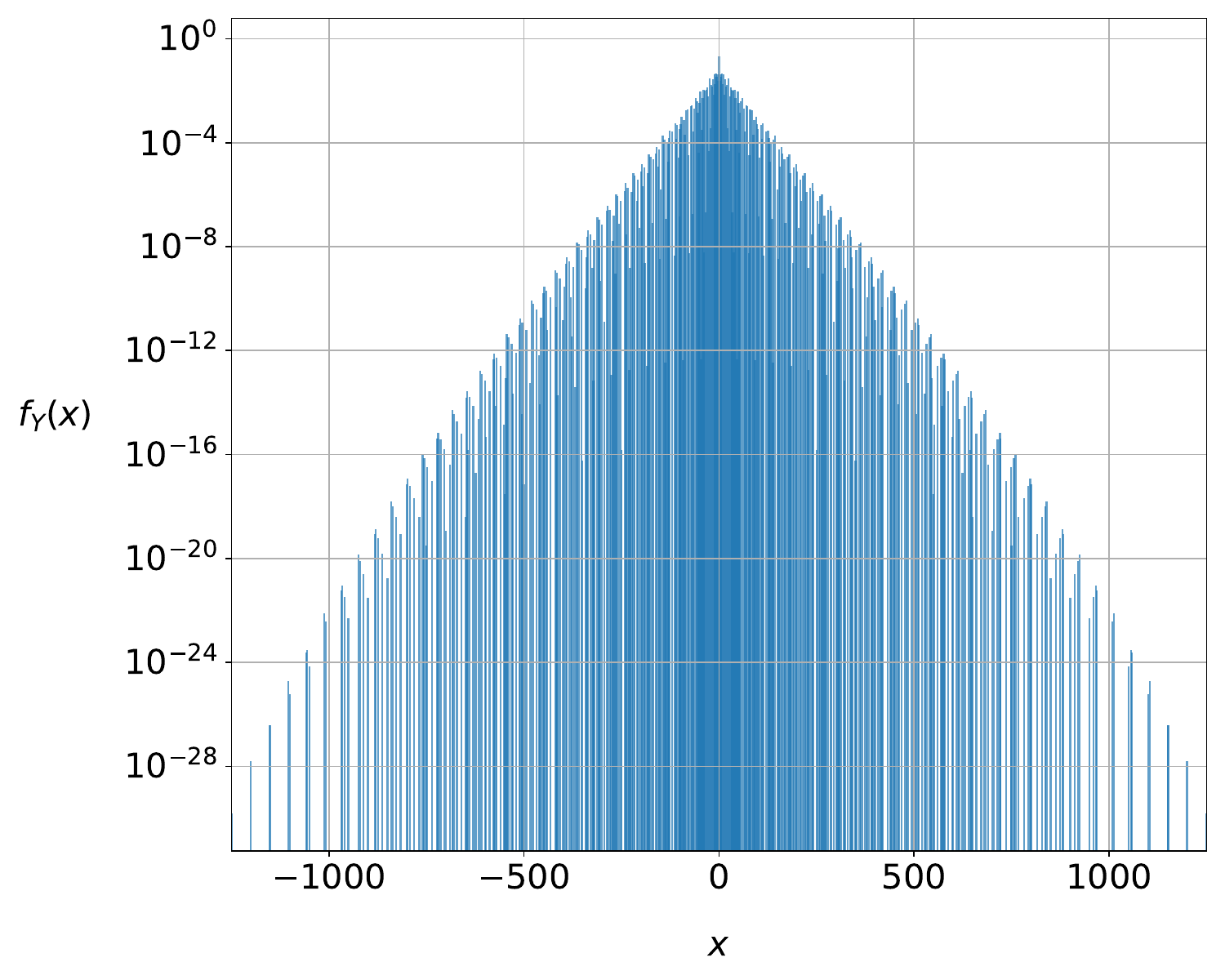}} 
    \subfigure[]{\includegraphics[width=0.49\textwidth]{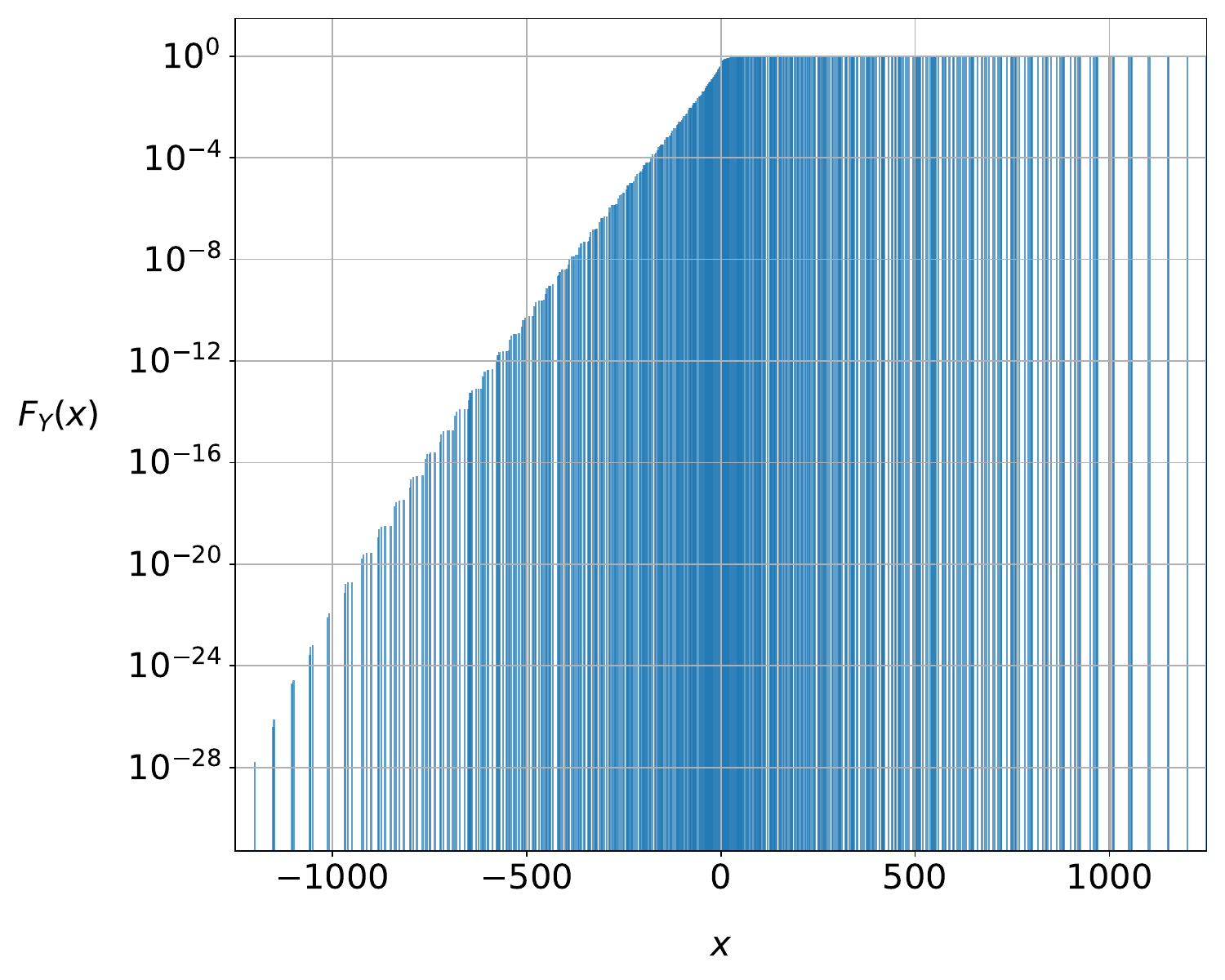}} 
    \caption{(a) Probability mass function and (b) cumulative distribution function concerning the random variable $Y$ for the Max-Cut problem on the graph $K_{50, 50}$. We compute efficiently the explicit distribution from our characterization of the solution space. Both graphs are on a log-linear scale. As the possible values for the cost are given by $\frac{1}{2} (n - 2j) (n - 2k)$ for all $0 \leq j, k \leq n$, the region near to $0$ is denser than the regions near to $R_{X_c}^{min}$ and $R_{X_c}^{max}$. Although the decay is not uniform, the general trend is clearly exponential.} 
    \label{maxcut_dist}
\end{figure*}

\begin{figure*}[ht]
    \centering
    \subfigure[]{\includegraphics[width=0.32\textwidth]{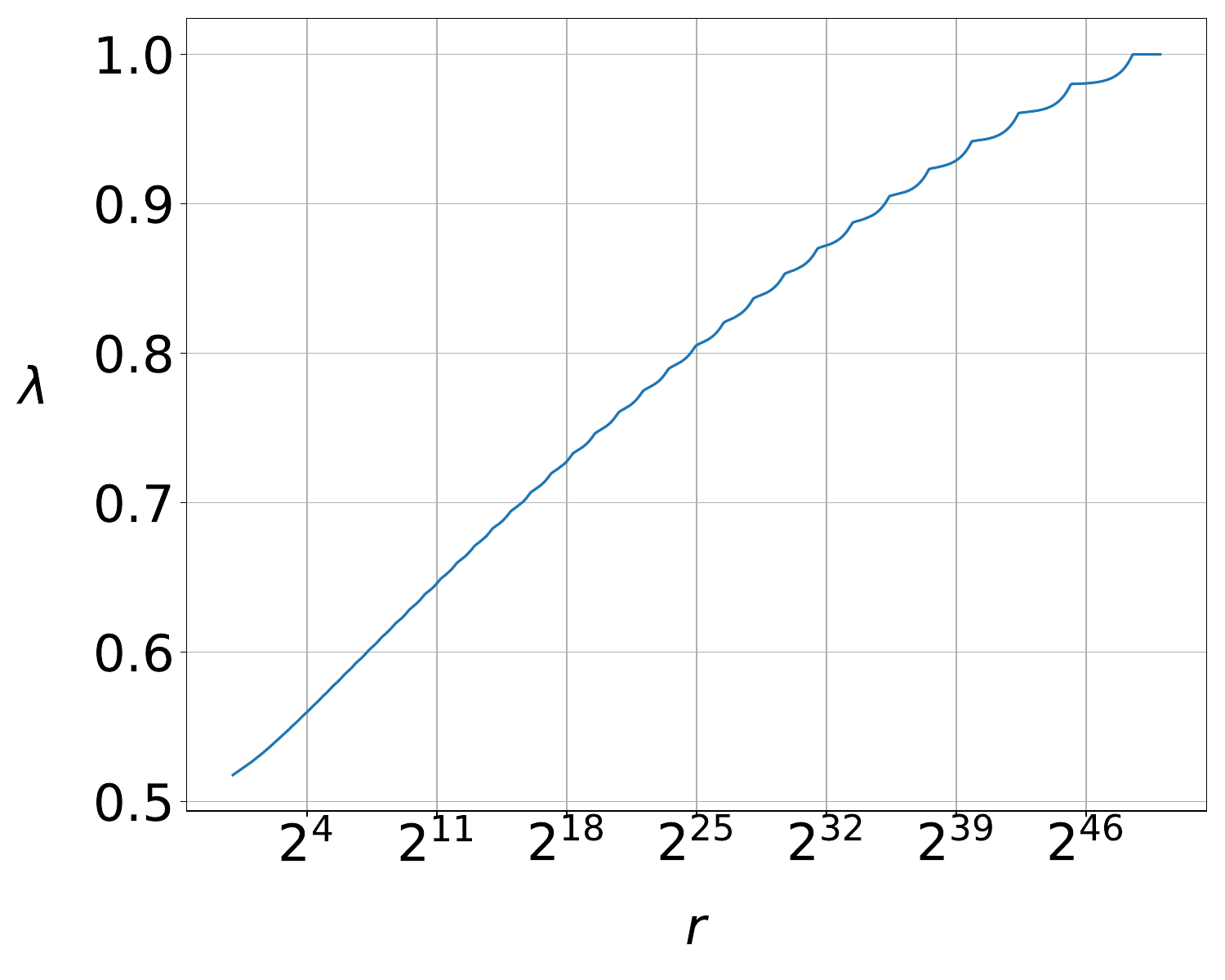}} 
    \subfigure[]{\includegraphics[width=0.32\textwidth]{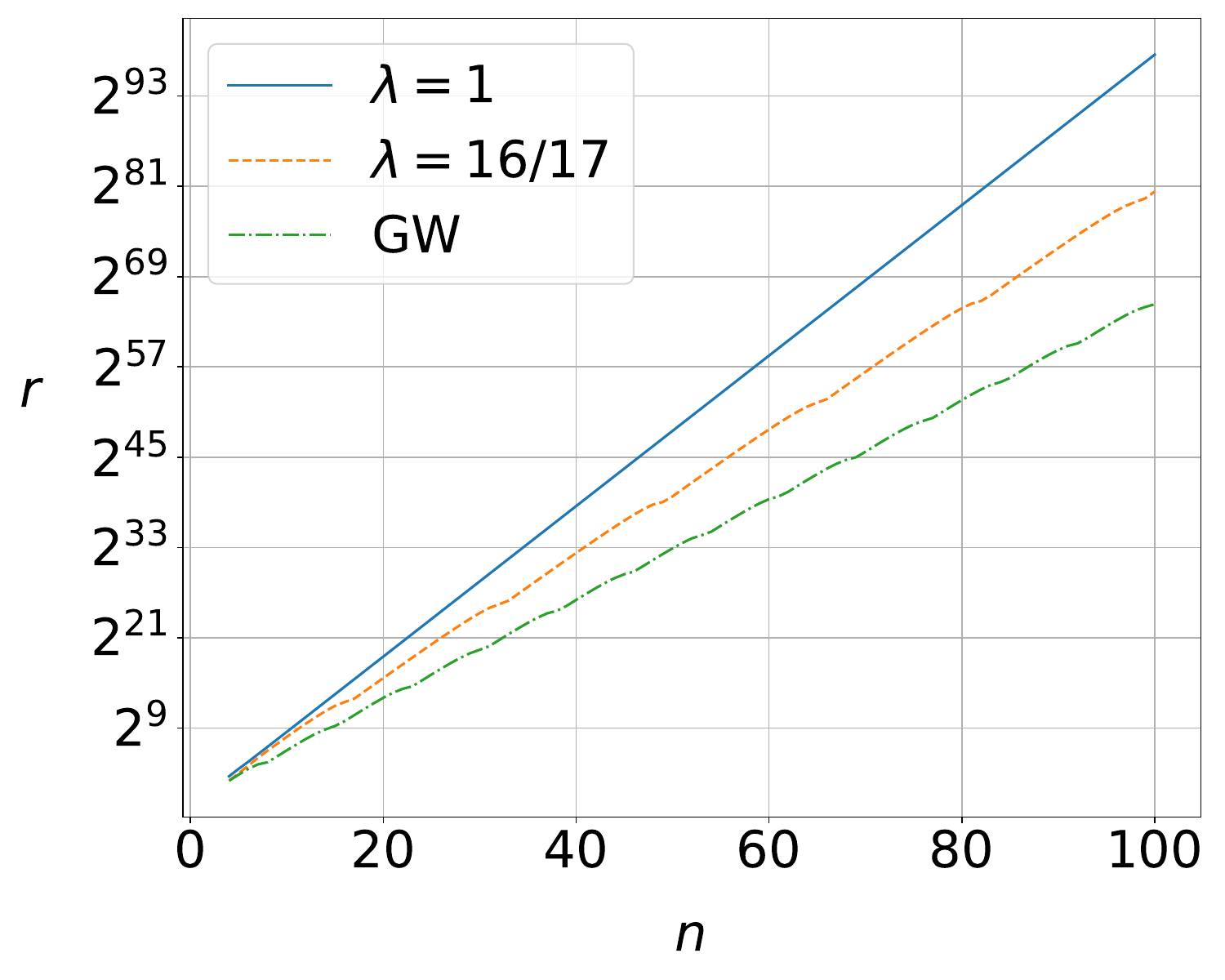}} 
    \subfigure[]{\includegraphics[width=0.32\textwidth]{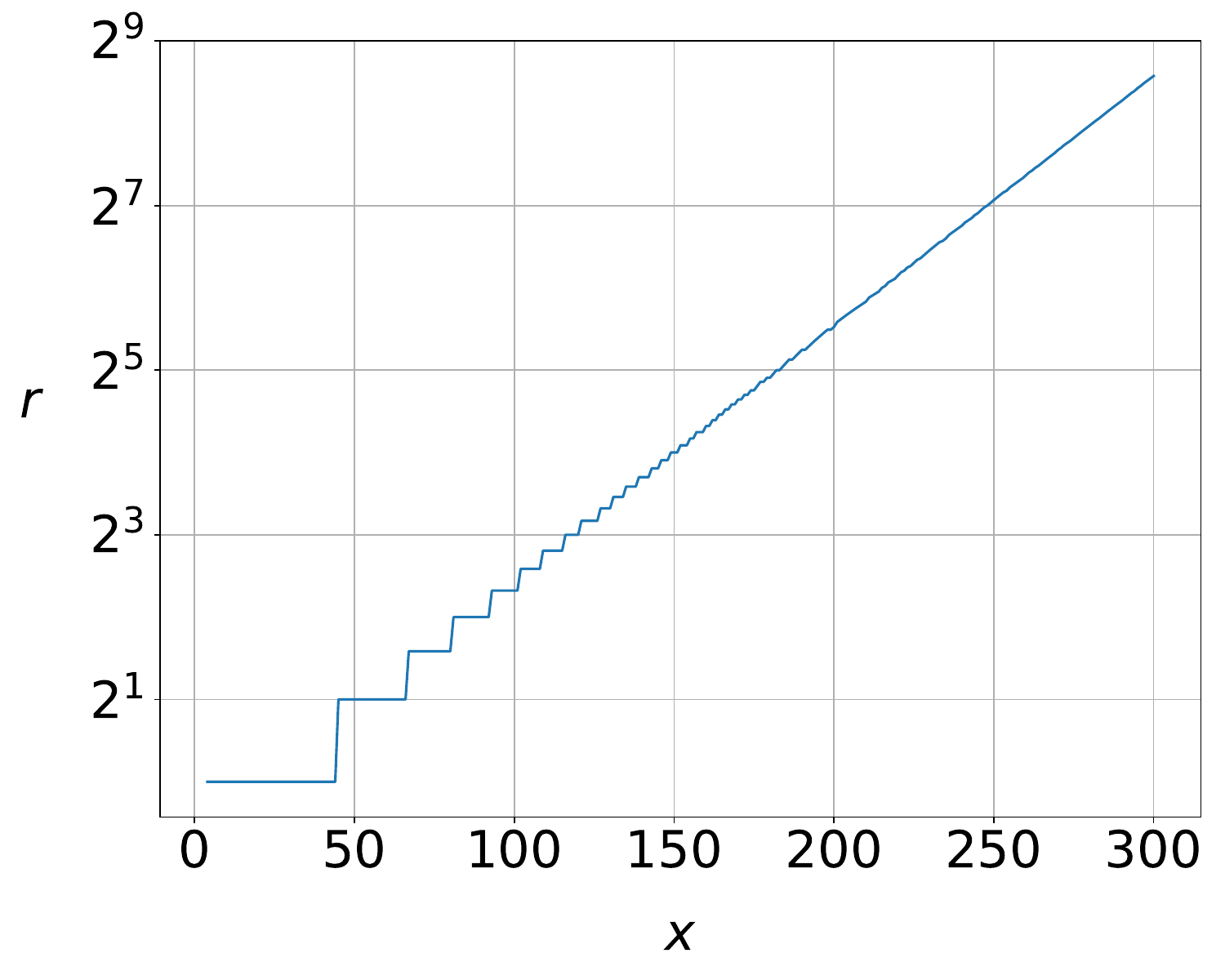}} 
    \caption{(a) Linear-log base $2$ plot of approximation ratio versus the number of layers considering the application of the maximum amplification bound for the graph $K_{50, 50}$ on Max-Cut. As expected, the growth rate is logarithmic. The resolution used on the number of layers (due to the extremely high number of layers to get $\lambda = 1$) is $\lceil 2^{x/100} \rceil$ for $x = 0, 1, \ldots, 5000$. (b) The minimum number of layers required to maximum amplification bound achieves three different values of approximation ratio on the Max-Cut with the graph $K_{n, n}$ for $n = 4, 5, \ldots, 100$. The considered values of $\lambda$ are $\lambda = 1$, in which we calculate analytically that $r = \lceil 2^{0.5(|\mathcal{V}| - 3)} \rceil$; $\lambda = 16/17$, the $\lambda$ value in which Max-Cut becomes NP-hard; and the approximation ratio guaranteed by the classical Goemans-Williamson algorithm, given by $\lambda \approx 0.8786$. The scale is log-linear with base $2$. The value of $r$ at which the approximation ratio is achieved was efficiently found with a binary search. As predicted, the number of layers scales exponentially in all of them. (c) The same as (b), but with $n = 4, 5, \ldots, 300$ and $\lambda = 0.52$. As the expectation value of a uniform superposition gives $\lambda = 0.5$, this approximation ratio is extremely low. However, even for such low performance, given a sufficient number of vertices, we observe the exponential dependence on the number of layers $r$.}
    \label{maxcut_sim}
\end{figure*}

In the specific context of GM-QAOA, Benchasattabuse et al.~\cite{lowerBounds}, by extending a lower bound on the time of quantum annealing to the context of QAOA~\cite{quantum_annealing}, prove that for the class of bipartite graphs, to achieve any constant approximation ratio guarantee, that variant requires a number of layers on the order of $\Omega(\sqrt{|\mathcal{E}|})$, where $|\mathcal{E}|$ is the number of edges of the graph. For the more general context of Grover-based QAOA, from the application of Theorem~\ref{thm6}, we improve by a constant factor that lower bound---given by Equation (36) of Benchasattabuse et al.~\cite{lowerBounds} paper---with
\begin{equation}
    \label{eqEa01}
    r \geq \frac{2\lambda - 1}{2\sqrt{1 + 1/r}} \sqrt{|\mathcal{E}|}.
\end{equation} 
In other words, for any choice of approximation ratio $\lambda$, the number of layers for the result of the algorithm reaches $\lambda$ must grow at least in order of $\sqrt{\mathcal{E}}$. Therefore, we cannot guarantee any fixed approximation ratio with a constant number of layers, a severe limitation for the NISQ devices, which require shallow depth. In general, from Theorem~\ref{thm6}, this problem will appear if the statistical quantity $(\mu - R_{X_c}^{min})/\sigma$ grows with the size of the instance to a given combinatorial optimization problem. However, at least in Max-Cut, the situation seems to be even worse, as we can see applying the bound of Eq.~\eqref{eqE09}. 

In that case, we consider additionally that the bipartite graph is connected. Since these graphs have a unique bipartition, the number of cuts of maximum size is $2$, and thus $f_{X_c}(R_{X_c}^{min}) = 1/2^{|\mathcal{V}| - 1}$, where $|\mathcal{V}|$ is the number of vertices of the graph. Therefore, $r$ scales exponentially on the numbers of vertices with $r \geq 2^{0.5(|\mathcal{V}| - 3)} = \Omega(\sqrt{2^{|\mathcal{V}|}})$ as $|\mathcal{V}| \rightarrow \infty$, and since $|\mathcal{E}| \geq |\mathcal{V}| - 1$ on connected graphs, also scales exponentially on $|\mathcal{E}|$. Of course, the bound is not applicable on approximate solutions with $ \lambda < 1$. Nevertheless, at least for the class of the complete bipartite graphs, we can argue that the growth is exponential by analyzing its probability distribution. 

Let us consider the complete bipartite graph $K_{n, n}$ with bipartition on the sets $V_1$ and $V_2$. Suppose a solution of Max-Cut as a partition on the sets $S_1$ and $S_2$ in which among the vertices of $S_1$, the number of vertices that belong to $V_1$ and $V_2$ are respectively $j$ and $k$. The size of the cut can be computed by discounting to the total number of edges $n^2$ of the whole graph, $jk + (n - j)(n - k)$, i.e., the number of edges induced by the union of both complete bipartite graphs $K_{j, k}$ and $K_{n - j, n - k}$ within the sets $S_1$ and $S_2$, respectively. Thus, with our definition of considering minimization problems and using the random variable $Y$ by subtraction of the mean $-n^2/2$, the solution space is composed by $\binom{n}{j} \binom{n}{k}$ solutions of cost $\frac{1}{2} (n - 2j) (n - 2k)$ for all $0 \leq j, k \leq n$. Although we do not have the explicit distribution, the current characterization of the solution space is sufficient to see that it presents an exponential decay toward the optimal solution. As the distribution is symmetric, we consider, by simplicity, $j, k \leq \lfloor n/2 \rfloor$. Fixing $j$ (or $k$), with a large enough $n$, the number of solutions grows exponentially with $k$ (or $j$) at the same time that the costs have a linear decrease. That combined behavior implies that a trend of a linear increase of the costs to lower values of $k$ and $j$ is accompanied by an exponential decay of their probabilities. Fig.~\ref{maxcut_dist} illustrate the decay of the distribution for the graph $K_{50, 50}$ by showing the graphics of $f_Y(x)$ and $F_Y(x)$. The exponential decay, with the arguments of Sec.~\ref{sub3D} on the asymptotic limit, infers in a logarithmic increase of $C_r$ with the number of layers, and therefore, as $(\mu - R_{X_c}^{min})/\sigma$ grows with the square root on the number of edges for bipartite graphs, the number of layers to achieve guaranteed in terms of $\lambda$ must increase exponentially with the number of vertices/edges. 

That result points out a limitation of Grover-based QAOA that is serious beyond the NISQ context since the exponential increase in the number of rounds implies that we cannot obtain a polynomial-time algorithm that guarantees any constant approximation ratio. In the context of comparative literature between QAOA mixers, our findings for QAOA with Grover mixer contrast with the aforementioned result that the QAOA with transverse field mixer guarantees a constant approximation ratio for $3$-regular graphs even with a single round.  

To illustrate our arguments, we simulate the maximum amplification bound on the complete bipartite graphs of different sizes. Fig.~\ref{maxcut_sim}(a) shows the logarithm growth of the approximation ratio on $r$ to the graph $K_{50, 50}$. Fig.~\ref{maxcut_sim}(b) and Fig.~\ref{maxcut_sim}(c) display the exponential dependence on the number of rounds to achieve different values of approximation ratio when we scale $n$. In the first between the last two figures mentioned, we choose approximation ratios with practical interest, while in the last one, we consider an extremely low approximation ratio to emphasize the stiffness of the limitation.

In a more general context, a given type of instance must suffer from the same limitation if there is simultaneously a distribution with exponential decay and the quantity $(\mu - R_{X_c}^{min})/\sigma$ grows above the logarithmic rate with the size of the problem. Thus, it is likely that there are other classes of graphs on Max-Cut and other types of instances beyond the Max-Cut problem that fit into these conditions. For example, we can mention the aforementioned normally distributed instances of the Capacitated Vehicle Routing and Portfolio Optimization problems~\cite{vehicle, maoa}, which meet the first criterion.

\section{Conclusions} \label{sec6}
In the present work, we apply the statistical approach of Headley and Wilhelm~\cite{headley_paper} on GM-Th-QAOA, obtaining an expression for the expectation value with complexity independent of the number of layers. With the expression, we first solve the conjecture of the threshold curve and then get bounds of different natures, including on the statistical quantities of quantile and the standard score, and on the minimum number of layers required for the algorithm to guarantee some fixed approximation ratio. The bound on the quantile is of particular interest since it reflects explicitly a quadratic Grover-style speed-up. Subsequently, we generalize the GM-Th-QAOA bounds to the general Grover-based QAOA framework. We derive these findings from the result---established through a proof by contradiction with the optimality of the unstructured search problem---that the probability of measuring states on Grover-based QAOA is bounded by a quadratic growth on the number of layers. The findings showed that Grover-based QAOA achieves, at most, the same asymptotic performance as GM-Th-QAOA. Consequently, we obtain the main contribution of this work: the formal establishment that the Grover mixer’s performance is bounded by the quadratic bound of the unstructured search problem, confirming previous conjectures in the literature~\cite{numericalEvidence1, numericalEvidence2}. That limiting can be severe for combinatorial optimization, as evidenced by the application of Max-Cut on the complete bipartite graph, which we conclude by an argument using the asymptotic decay rate of the probability distribution that it requires an exponential number of layers to maintain constant performance and, therefore, we cannot have a polynomial-time algorithm that guarantees a fixed approximation ratio. This highlights that in order to get significant results with QAOA, especially in the NISQ era, it is essential for the algorithm to explore the structure of the optimization problems. Indeed, recall the numerical evidence of Golden et al.~\cite{numericalEvidence2} that suggests the possibility of exponential gain of QAOA with structure-dependent mixer over Grover mixer variants. Thus, research should be directed toward understanding the mechanisms by which different types of mixers can benefit from the structure of particular problems, a path opened by Headley~\cite{headley_thesis} with the statistical approach on the transverse field mixer and the line mixer. 

For a more accurate picture of the limitation of QAOA with Grover mixer, future works would consider the application of the maximum amplification bound to more graph classes on Max-Cut and other combinatorial optimization problems, and so estimate how common the need for exponential growth on the layers to achieves a target approximation ratio. Fortunately, explicit knowledge of the distribution is not necessarily required to find the asymptotic behavior and establish how the performance scales---as was the case of complete bipartite graphs on Max-Cut---which could make the task much easier. A case of particular interest is the investigation of classes of instances where QAOA with another mixer besides Grover mixer has some known guarantee on the approximation ratio, such as Max-Cut on $3$-regular graphs~\cite{qaoa, maxcut_ar1} and Sherrington-Kirkpatrick model at infinity size~\cite{sherrington-kirkpatrick}, both for traverse field mixer. If the maximum amplification bound requires exponential resources to maintain constant performance for a case of this type, the direct comparison would underscore the limitation of the Grover mixer as a mixer for QAOA.

Yet in the Grover mixer context, there are at least two open questions. Firstly, it would be interesting to decide whether GM-Th-QAOA is the best Grover-based QAOA for all possible instances, or at least whether GM-Th-QAOA outperforms GM-QAOA always, confirming the numerical evidence. Intuitively, it is reasonable to think that the most efficient agnostic-structure method possible is to compile the cost function on a binary function and perform Grover's algorithm. The results and insights of the present work indicate that this can be the case. However, formal proof is still needed. Secondly, one can ask whether GM-QAOA even reflected the quadratic Grover-style speed-up in the sense of Theorem~\ref{thm3}. However, insights would be needed to answer that question analytically since direct analytical treatment to bound Eq.~\eqref{eqCc03} is infeasible.

\begin{acknowledgments}
G.A.B. thanks the financial support of CAPES/Brazil. F.L.M. thanks the financial support of CNPq/Brazil grant 407296/2021-2, and Latvian Quantum Initiative project 2.3.1.1.i.0/1/22/I/CFLA/001. The authors thank R.~Portugal and F.~Santos for helpful discussions.
\end{acknowledgments}

\bibliographystyle{apsrev4-2}
\bibliography{apssamp}% Produces the bibliography via BibTeX.

\appendix

\section{Proof of Theorem~\ref{thm2}} \label{ap1}

\begin{widetext}
Firstly, note that if $F_Y(T) > \rho_{Th}(r)$, by Eq.~\eqref{eqC14}, $E_r(t)$ is monotonically non-decreasing and therefore since $P(\rho, r)$ is continuous, it is enough to prove that the monotonicity change at most one time on $F_Y(T) \leq \rho_{Th}(r)$ interval. To get it, using $\rho = F_Y(T)$ to simplify notation, the derivative can be computed as
\begin{equation}
    \label{eqApA01}
    \frac{d E_r(t)}{dT} = f_X^G(T) \left(- \frac{T (1 - \rho) + G_Y(T)}{(1 - \rho)^2} + T\frac{\eta}{1 - \rho} + G_Y(T) \frac{\eta_o(1 - \rho) + \eta}{(1-\rho)^2} \right),
\end{equation}
where $\frac{d\eta}{dT} = f_X^G(T) \eta_o$. By definition, $f_X^G(T)$ is non-negative, and then we can ignore it. Moreover, we can also ignore the points in which $\rho = 0$ since the expectation value is $\mu$. We define $R = 2r + 1$ and $D_1 = \sqrt{\rho(1 - \rho)}/(R \arcsin{(\sqrt{\rho})})$. Multiplying the expression by the convenient positive factor $(1-\rho) D_1/\eta$, we get
\begin{equation}
    \label{eqApA02}
    T \left(1 - \frac{1}{\eta} \right) D_1 - G_Y(T) \frac{1}{\eta} D_2 + G_Y(T) D,
\end{equation}
where
\begin{equation}
    \label{eqApA03}
    D = \left( \frac{\eta_o}{\eta} + \frac{1}{1 - \rho} \right) D_1 , \ D_2 = \frac{D_1}{1 - \rho}.
\end{equation}
At $T \rightarrow -\infty$, the derivative begins negative, which follows from $\eta > 1$, $G_Y(T) = 0$, and $D_1 > 0$. Furthermore, we demonstrate further that $D$ is negative. That way, if $T \geq 0$, the expression does not change the sign since all terms are non-negative. So, we assume $T < 0$. As $G_Y(T)$ is non-increasing and $\eta > 1$, if we prove that (i) $\eta$ is strictly decreasing (ii) $D_1$ is strictly decreasing, (iii) $D_2$ is strictly increasing, and (iv) $D$ is strictly decreasing and negative, all terms of Eq.~\eqref{eqApA02} are non-decreasing and the monotonicity of $E_r(t)$ change one time on this interval, proving the theorem. The minimum of the original discrete function is hit either on $\rho \leq \rho_{Th}(r)$ or in the smallest defined $F_Y(T)$ in which $\rho > \rho_{Th}(r)$.

Consider the substitution $u = R \arcsin{(\sqrt{\rho})}$. By the chain rule, since $\frac{du}{d\rho}$ is positive for all $\rho \leq \rho_{Th}(r)$, we can analyze directly the derivative with respect to $u$. The equivalent interval of $u$ is $0 < u \leq \pi/2$. Thus,
\begin{equation}
    \label{eqApA04} 
    \eta = \left(\frac{\sin{(u)}}{\sin{(u/R)}}\right)^2 , \ D_1 = \frac{\sin{(2u/R)}}{2u}, \ D_2 = \frac{\tan{(u/R)}}{u}, \ D = \frac{R \cot{(u)} - \cot{(u/R)}}{u} + \frac{\tan{(u/R)}}{u}.
\end{equation}

The proofs of the first three cases are omitted since they can be established directly by using derivatives. The last one, on the other hand, is more complicated. The limit of $D$ on $u \rightarrow 0$ can be computed with the Taylor series of $\sin{(x)}$, $\cos{(x)}$ and $1/(1 - x)$, resulting in $(4 - R^2)/(3 R)$, which is negative for any $r$. Therefore, demonstrating that $D$ is strictly decreasing over the entire interval implies that it is also negative. Then, taking the derivative and ignoring the positive factor $Ru^2$, we must satisfy
\begin{equation}
    \label{eqApA05}
  R \cot{(u/R)} - R^2 \cot{(u)} + u \csc^2{(u/R)} - R^2 u \csc^2{(u)} + u \sec^2{(u/R)} - R \tan{(u/R)} < 0.
\end{equation}
We can conclude in an analogous way as $D$ that the limit of the left side of Eq.~\eqref{eqApA05} is $0$. Considering the stronger condition that its derivative is also negative, with the trigonometric identity $\cos{(2x)} = \cos^2{(x)} - \sin^2{(x)}$ we can get
\begin{equation}
    \label{eqApA06}
       2u R^2 \left( \frac{\cos{(u)}}{\sin^3{(u)}} - \frac{8 \cos{(2u/R)}}{R^3 \sin^3{(2 u/R)}} \right) < 0   \Rightarrow \frac{\sin^3{(u)}}{\cos{(u)}} - \frac{R^3 \sin^3{(2 u/R)}}{8 \cos{(2u/R)}} > 0.
\end{equation}
Repeating the same argument, we finally have the condition
\begin{equation}
    \label{eqApA07}
    \tan^2{(u)} + 2 \sin^2{(u)} > \frac{R^2}{4}( \tan^2{(2u/R)} + 2 \sin^2{(2u/R)}).
\end{equation}

Now, consider the Taylor series expansion of $\tan^2{(x)}$ and $2 \sin^2{(x)}$. Using the trigonometric relations $\tan^2{(x)} + 1 = \frac{d\tan{(x)}}{dx}$ and $2\sin^2{(x)} = 1 - \cos{(2x)}$, we can conclude that
\begin{equation}
    \label{eqApA08}
     \tan^2{(x)} = \sum_{n = 1}^\infty 
    \frac{B_{2n + 2}(-4)^{n+1} (1 - 4^{n+1})}{(2n + 2)(2n)!} x^{2n}, \ 2 \sin^2{(x)} = \sum_{n = 1}^\infty \frac{(-1)^{n+1} 2^{2n}}{(2n)!} x^{2n},
\end{equation}
where $B_{2n}$ is the Bernoulli number. Note that both expansions are non-zero for terms with the same order, which allows a direct comparison. We are interested in establishing that all terms of $\tan^2{(x)} + 2 \sin^2{(x)}$ expansion are non-negative. The sine squared expansion alternates the sign, having a negative sign for $n$ even, while tangent squared expansion has all positive terms. Then, it is enough to show that all even $n$ terms of $\tan^2{(x)}$ expansion are equal or greater than the absolute value of the respective terms of $2 \sin^2{(x)}$. Comparing both expansions on Eq.~\eqref{eqApA08}, for $n = 2$, they are the same and followed by an induction argument that our claim holds for $n > 2$. 

To finish, we demonstrate that each expansion term on the left side of Eq.~\eqref{eqApA07} is equal to or larger than its counterpart on the expansion of the right side. To get that, consider an arbitrary term $n$ of the expansions. The argument of trigonometric functions on the left side is $x = u$ while the left side is $x = 2u/R$. Combining it with the constant multiplication factor $R^2/4$ on the right size, the term of the right diverges from the left by the factor $(2/R)^{2n-2}$. They are equal for $n = 1$, and as $R \geq 3$, the left side is larger for $n > 1$. Hence, since we prove that all expansion terms are non-negative, the left side is larger than the right, and the inequality of Eq.~\eqref{eqApA07} holds, establishing the theorem.

\section{Proof of Lemma~\ref{lm1}} \label{ap2}

By Eq.~\eqref{eqC05}, for a fixed $F_Y(T)$, the maximum possible $C_r(t)$ is given when we maximize the ratio $|G_Y(T)|/\sigma$. Eliminating the trivial cases of $F_Y(T) = 0$ and $F_Y(T) = 1$, the key idea of the proof is that the distribution that maximizes that ratio is a two-point distribution for all the range $0 < F_Y(T) < 1$.

To get it, we split the distribution into two parts by the threshold value defining $Y_{\leq T}$ as the random variable $Y$ given $Y \leq T$ and $Y_{> T}$ as the random variable $Y$ given $Y > T$. We also split the summation that computes $\sigma^2$ into two contributions
\begin{equation}
    \label{eqApB01}
\sigma^2_{\leq T} = \sum_{x \in R_Y: x \leq T} x^2 f_Y(x) , \ \sigma^2_{> T} = \sum_{x \in R_Y: x > T} x^2 f_Y(x),
\end{equation}
where $\sigma = \sqrt{\sigma^2_{\leq T} + \sigma^2_{> T}}$. Thus, using the definitions of the conditional random variables, 
\begin{equation}
    \label{eqApB02}
 \operatorname{E}[Y_{\leq T}] = \frac{G_Y(T)}{F_Y(T)}, \ \operatorname{E}[Y_{\leq T}^2] = \frac{\sigma^2_{\leq T}}{F_Y(T)}, \ \operatorname{E}[Y_{> T}] = -\frac{G_Y(T)}{1 - F_Y(T)}, \ \operatorname{E}[Y_{> T}^2] = \frac{\sigma^2_{> T}}{1 - F_Y(T)},
\end{equation}
to compute the bounds $\operatorname{E}[Y_{\leq T}^2] \geq \operatorname{E}[Y_{\leq T}]^2$ and $\operatorname{E}[Y_{< T}^2] \geq \operatorname{E}[Y_{< T}]^2$---that follow from Jensen's inequality---we have
\begin{equation}
    \label{eqApB03}
\frac{G_Y(T)^2}{F_Y(T)} \leq \sigma^2_{\leq T}, \ \frac{G_Y(T)^2}{1 - F_Y(T)} \leq \sigma^2_{> T}.
\end{equation}
Combining both bounds gives
\begin{equation}
    \label{eqApB04}
     \frac{|G_Y(T)|}{\sigma} \leq \sqrt{F_Y(T)(1 - F_Y(T))},
\end{equation}
and the equality is hit on the limit of Jensen's inequality if and only if both $Y_{\leq T}$ and $Y_{> T}$ are degenerate distributions.

\section{Proof of Lemma~\ref{lm2}} \label{ap3}

The expression $\eta_r(S_T, f_{X_q})$ for arbitrary depth is given by
\begin{equation}
    \label{eqApC01}
       \eta_r(S_T, f_{X_q}) = \sum_{k_\text{bra}, k_\text{ket} = 0}^{2^r - 1}  \exp{\left(i x_o \sum_{j \in P_\text{central}} \gamma_j\right)}  \prod_{P \in P_{bra}} \varphi_{X_q} \left( \sum_{j \in P} \gamma_j \right) \prod_{P \in P_\text{ket}} \varphi_{X_q} \left( \sum_{j \in P} \gamma_j \right) \prod_{-j: k_\text{bra}^j = 1}  B(\beta_j)  \prod_{j: k_\text{ket}^j = 1} B(\beta_j).
\end{equation}
To simplify the notation in this proof, we hide the subscript with the random variable on the probability distribution and the characteristic function. In contrast, we distinguish between various probability distributions, their characteristic functions, and specific components of their summations through superscripts, without defining explicitly random variables. We also say the phase separation operator \textit{computes} the probability distribution $f$ if the distribution associated with $q(k)$ is $f$. Furthermore, to compact the notation of Eq.~\eqref{eqApC01}, we group under the notation $\Phi(\varphi, N_\varphi, k)$ both products involving characteristic functions, and we group under the notation $\mathcal{B}(N_B, k)$ the exponential factor as well as both products involving $B(\beta)$. Thus,
\begin{equation}
    \label{eqApC02}
      \eta_r(S_T, f) = 
     \sum_{k_\text{bra}, k_\text{ket} = 0}^{2^r - 1} \mathcal{B}(N_B, k) \Phi(\varphi, N_\varphi, k),
\end{equation}
where $k$ is the ordered pair $(k_\text{bra}, k_\text{ket})$, $N_\varphi$ is the number of characteristic functions factors, and $N_B$ is the number of $B(\beta)$ factors.

Suppose by contradiction that for some $\epsilon > 0$, there is a choice of distribution $f^\text{O}(x)$ (original distribution) in which $\eta_r(S_T, f^\text{O}) = (2r + 1)^2 + \epsilon$. We fix the optimal variational parameters. Moreover, we can set $x_o = 1$ by a shifting location without loss of generality. Note that we can express the characteristic function of $f^\text{O}(x)$ as 
\begin{equation}
    \label{eqApC03}
     \varphi^\text{O}(\gamma) = \rho e^{i\gamma} + \varphi^{\text{rem}}(\gamma),
\end{equation}
where the first term represents the portion of the summation of the characteristic function for $S_T$ with ratio $\rho$, and $\varphi^{rem}(\gamma)$ is the remainder portion. Let $\delta$ be a rational such that $0 < \delta \leq 2 \rho$. We can rewrite $\varphi^{\text{O}}$ as 
\begin{equation}
    \label{eqApC04}
     \varphi^\text{O}(\gamma) = 0.5 \delta e^{i\gamma} + \varphi^\text{R}(\gamma),
\end{equation}
where $\varphi^\text{R}(\gamma) = (\rho - 0.5 \delta) e^{i\gamma} + \varphi^{\text{rem}}(\gamma)$ and $0.5 \delta e^{i\gamma}$ represents the portion of the summation for a subset $S_\delta$ of $S_T$. Since Grover-based QAOA preserves the equality of amplitudes in degenerate states during the unitary evolution, $\eta_r(S_\delta, f^\text{O}) = \eta_r(S_T, f^\text{O})$. 

Consider the following algorithm for the unstructured search problem with $m$ marked elements over $M$ solutions with $0.5 \delta = m/M$ and $r$ rounds. The $k$th diffusion operator is the sequential application of a phase separation operator that computes a target distribution $f^\text{T}(x)$ and the Grover mixer operator. Both with the fixed parameters of the $k$th layer of Grover-based QAOA. The distribution $f^\text{T}(x)$ has the characteristic function
\begin{equation}
    \label{eqApC05}
     \varphi^\text{T}(\gamma) = e^{i\theta \gamma} (\varphi^\text{rea}(\gamma) + \varphi^\text{R}(\gamma)).
\end{equation}
The term $\varphi^\text{rea}(\gamma)$ represents an arbitrary reassignment of the costs of the marked elements of the original distribution replacing $0.5\delta e^{i \gamma}$ and the factor $e^{i\theta \gamma}$ is a location shift of the distribution of size $\theta > 0$. The quantity $m$ can be chosen as the minimum number of marked elements required for computing the distribution $f^\text{T}(x)$.

On the other hand, for marked elements, the oracle reverses the action of the defined phase separation and then applies a phase shift of $e^{i \gamma}$, i.e., a mapping on the cost $1$. In practice, the oracle interrupts the computation of target distribution $f^\text{T}(x)$ of the phase separation on specific marked values in such a way that the combined action of diffusion and oracle operators encodes the value $1$. Consequently, the algorithm's performance depends on the positions of the marked elements, and thereby the average probability is unknown. For using the optimality of Hamann, Dunjko, and Wölk~\cite{generalOptimalGrover} on the search problem, we bound the minimum probability value in terms of the known performance of original distribution $f^\text{O}(x)$ by choosing values of $\delta$ and $\theta$ sufficiently small. 

Let $f^\text{sp}(x)$ (search problem) be a distribution computed by the combined action of the phase separation and the oracle for an arbitrary instance of the search problem algorithm. The characteristic function of $f^\text{sp}(x)$ can be expressed without loss of generality by
\begin{equation}
    \label{eqApC06}
     \varphi^\text{sp}(\gamma) = 0.5 \delta e^{i\gamma} + e^{i\theta \gamma} \varphi^\text{R}(\gamma) + \varphi^{\text{1}}(\gamma) - \varphi^{\text{2}}(\gamma),
\end{equation}
where $0.5 \delta e^{i\gamma}$ represents the oracle finding the marked elements; $e^{i\theta \gamma} \varphi^\text{R}(\gamma)$ is the portion of $f^\text{T}(x)$ computed up to the phase shifting $e^{i\theta \gamma}$ on the original characteristic function $\varphi^\text{O}(x)$; $\varphi^{\text{1}}(\gamma)$ represents the non-computed part of $f^\text{T}(x)$ on the original distribution $f^\text{O}(x)$ that is computed on distribution $f^\text{sp}(x)$; and $\varphi^{\text{2}}(\gamma)$, in contrast, represents the non-computed part of $f^\text{T}(x)$ on the distribution $f^\text{sp}(x)$ that is computed on $f^\text{O}(x)$. We denote 
\begin{equation}
    \label{eqApC07}
     \varphi^\text{eq}(\gamma) = 0.5 \delta e^{i\gamma} + e^{i\theta \gamma} \varphi^\text{R}(\gamma) , \ \varphi^\text{dif}(\gamma) = \varphi^{\text{1}}(\gamma) - \varphi^{\text{2}}(\gamma).
\end{equation}
The first definition, $\varphi^\text{eq}(\gamma)$ (equal), represents the original distribution $f^\text{O}(x)$ up to the location shifting on $f^\text{R}(x)$. The second definition, $\varphi^\text{dif}(\gamma)$ (different), represents the divergence between positions of $f^\text{O}(x)$ and $f^\text{sp}(x)$. If the marked elements on both distributions are at the same positions, $\varphi^\text{dif}(\gamma) = 0$, and since at worst case it diverges on all their $2m$ marked elements, $|\varphi^\text{dif}(\gamma)| \leq \delta$. By Eqs.~\eqref{eqApC04} and~\eqref{eqApC07}, $\varphi^\text{eq}(\gamma)$ can be written as $\varphi^\text{eq}(\gamma) = \varphi^\text{O}(\gamma) + (e^{i\theta \gamma} - 1) \varphi^\text{R}(\gamma)$ and then we can write $\varphi^\text{sp}(\gamma)$ from Eq.~\eqref{eqApC06} in terms of $\varphi^\text{O}(\gamma)$ by
\begin{equation}
    \label{eqApC08}
     \varphi^\text{sp}(\gamma) = \varphi^\text{O}(\gamma) + (e^{i\theta \gamma} - 1) \varphi^\text{R}(\gamma) + \varphi^\text{dif}(\gamma).
\end{equation}
Using $\gamma^{(j)}$ as a generic notation for the argument of the $j$th characteristic function on the product (for any arbitrary order) and setting $\vartheta = \theta \gamma_{max}$ with $\gamma_{max}$ being the maximum absolute value of an argument, we bound $\Phi(\varphi^\text{sp}, N_\varphi, k)$ by
\begin{equation}
    \label{eqApC09}
    \begin{split}
  \Phi(\varphi^\text{sp}, N_\varphi, k) & = \prod_{j = 1}^{N_\varphi} \varphi^\text{O}(\gamma^{(j)}) + ((e^{i\theta \gamma^{(j)}} - 1) \varphi^\text{R}(\gamma^{(j)}) + \varphi^\text{dif}(\gamma^{(j)})) \\ & =  \sum_{\bm{x}} \prod_{j = 1}^{N_\varphi} (\varphi^\text{O}(\gamma^{(j)}))^{1- x_j} ((e^{i\theta \gamma^{(j)}} - 1) \varphi^\text{R}(\gamma^{(j)}) + \varphi^\text{dif}(\gamma^{(j)}))^{x_j} \\ & \geq \Phi(\varphi^\text{O}, N_\varphi, k) - 2^{N_\varphi} (\delta + \vartheta),
    \end{split}
\end{equation}
where $\bm{x} = (x_1, \ldots , x_{N_\varphi})$ is a $N_\varphi$-bit string. The inequality follows from the individual bounds $|\varphi^\text{O}(\gamma^{(j)})| \leq 1$, $|\varphi^\text{R}(\gamma^{(j)})| \leq 1$,  $|\varphi^\text{dif}(\gamma^{(j)})| \leq \delta$, $|e^{i\theta \gamma^{(j)}} - 1| \leq \vartheta$. The last individual bound follows from $\cos{(x)} \geq 1 - x^2/2$. The maximum value of both $N_\varphi$ and $N_B$ are $2r$ and $|\mathcal{B}(N_B, k)|$ is bounded by $2^{N_B}$. Combining those results with Eq.~\eqref{eqApC02} gives
\begin{equation}
    \label{eqApC10}
      \eta_r(S_\delta, f^\text{sp}) \geq \eta_r(S_\delta, f^\text{O}) - 64^r (\delta + \vartheta) = (2r + 1)^2 + \epsilon - 64^r (\delta + \vartheta).
\end{equation}
For any $r$, there is a choice of $\delta$ and $\vartheta$ in which $\epsilon > 64^r (\delta + \vartheta)$. Combining it with the fact that the maximum amplification on Grover's algorithm is $(2r + 1)^2$, the optimality of Hamann, Dunjko, and Wölk~\cite{generalOptimalGrover} on the maximum average probability on the unstructured search problem with multiple marked elements is contradicted and establishes the lemma.

Xie et al.~ \cite{japanese}, after the first preprint version of our paper, introduced an alternative proof for Lemma~\ref{lm2} using mathematical optimization techniques to identify critical points through partial derivatives.

\section{Proof of Theorem~\ref{thm6}} \label{ap4}

In a similar way to the GM-Th-QAOA, the proof consists of concluding that the binary function achieves the maximum $C_r$, but now by slightly modifying the argument of Lemma~\ref{lm1}. In that case, we consider $\tau_1$ and $\tau_2$ from Theorem~\ref{thm5} and define the random variables $X_{1}$ and $X_{2}$, where the probability distribution $f_{X_1}(x)$ is given by $f_{X_1}(x) = f_{X_c}(x)$ for all $x \neq \tau_2 \in R_{X_C}$, $f_{X_1}(\tau_2) = 1/(2r+1)^2 - F_{X_c}(\tau_1)$, and the remainder probability to reach the summation of probabilities equal to $1$ can be arbitrarily assigned on values above $\tau_2$; and the probability distribution $f_{X_2}(x)$ is $f_{X_2}(x) = f_{X_c}(x)$ for all $x \neq \tau_2 \in R_{X_C}$, $f_{X_2}(\tau_2) = f_{X_c}(\tau_2) - f_{X_1}(\tau_2)$, with the remainder probability again arbitrarily assigned but for values below $\tau_2$. The random variable $X_1$ allows us to write the bound of Theorem~\ref{thm5} as
\begin{equation}
    \label{eqApD01}
E_r \geq \operatorname{E}[X_1| X_1 \leq \tau_2] = \frac{G_{X_1}(\tau_2)}{F_{X_1}(\tau_2)},
\end{equation}
with $F_{X_1}(\tau_2) = 1/(2r + 1)^2$. Now, we can apply the same proceeding of Lemma~\ref{lm1} on the random variables $X_{\leq \tau_2}$, the random variable $X_1$ given $X_1 \leq \tau_2$; and $X_{\geq \tau_2}$, the random variable $X_2$ given $X_2 \geq \tau_2$, to get the same conclusion. Then, take binary function on Eq.~\eqref{eqApD01} gives $C(r) \leq 2 \sqrt{r(r + 1)}$.

\end{widetext}

\nocite{*}

\end{document}